\documentclass[onefignum,onetabnum]{siamart171218}

\usepackage{amssymb,amsmath,verbatim,tabularx,enumitem,colonequals,graphicx,psfrag,subfigure,xcolor,algorithmicx,algpseudocode,MnSymbol}

\usepackage[algo2e,ruled,vlined,linesnumbered]{algorithm2e}

\usepackage[numbers,sort&compress]{natbib}

\usepackage{tikz}
\usepackage{nicefrac}

\usepackage[normalem]{ulem}
\usetikzlibrary{patterns}

\newcommand{\LB}{\textnormal{LB}}
\newcommand{\UB}{\textnormal{UB}}
\newcommand{\WS}{\textnormal{WS}}

\newcommand{\XE}{X_{\textnormal{E}}}
\newcommand{\XS}{X_{\textnormal{S}}}

\newcommand{\YN}{Y_{\textnormal{N}}}

\makeatletter
\pgfdeclarepatternformonly[\LineSpace]{my north east lines}{\pgfqpoint{-1pt}{-1pt}}{\pgfqpoint{\LineSpace}{\LineSpace}}{\pgfqpoint{\LineSpace}{\LineSpace}}%
{
    \pgfsetcolor{\tikz@pattern@color}
    \pgfsetlinewidth{0.4pt}
    \pgfpathmoveto{\pgfqpoint{0pt}{0pt}}
    \pgfpathlineto{\pgfqpoint{\LineSpace + 0.1pt}{\LineSpace + 0.1pt}}
    \pgfusepath{stroke}
}

\pgfdeclarepatternformonly[\LineSpace]{my north west lines}{\pgfqpoint{-1pt}{-1pt}}{\pgfqpoint{\LineSpace}{\LineSpace}}{\pgfqpoint{\LineSpace}{\LineSpace}}%
{
    \pgfsetcolor{\tikz@pattern@color}
    \pgfsetlinewidth{0.4pt}
    \pgfpathmoveto{\pgfqpoint{0pt}{\LineSpace}}
    \pgfpathlineto{\pgfqpoint{\LineSpace + 0.1pt}{-0.1pt}}
    \pgfusepath{stroke}
}
\makeatother

\newdimen\LineSpace
\tikzset{
    line space/.code={\LineSpace=#1},
    line space=3pt
}

\headers{The power of the weighted sum scalarization}{C. Bazgan, S. Ruzika, C. Thielen, and D. Vanderpooten}

\title{The power of the weighted sum scalarization for approximating multiobjective optimization problems\thanks{This work was supported by the DFG grants RU~1524/6-1 and TH~1852/4-1 as well as by a bilateral cooperation project funded by the German Academic Exchange Service (DAAD, Project-ID 57388848) and by Campus France, PHC PROCOPE 2018 (Project no. 40407WF).}}

\author{Cristina Bazgan\thanks{Universit\'e Paris-Dauphine, PSL Research University, CNRS,  LAMSADE, 75016 Paris, France 
  (\email{bazgan@lamsade.dauphine.fr}).}
\and Stefan Ruzika\thanks{Department of Mathematics, University of Kaiserslautern,
  Paul-Ehrlich-Str.~14, 67663~Kaiserslautern, Germany
  (\email{ruzika@mathematik.uni-kl.de}).}
\and Clemens Thielen\thanks{Technical University of Munich, TUM Campus Straubing for Biotechnology and Sustainability, Essigberg~3, 94315~Straubing, Germany
  (\email{clemens.thielen@tum.de}).}
 \and {Daniel Vanderpooten\thanks{Universit\'e Paris-Dauphine, PSL Research University, CNRS,  LAMSADE, 75016 Paris, France 
  (\email{daniel.vanderpooten@lamsade.dauphine.fr}).}}}

\begin{document}

\maketitle
    
\begin{abstract}
We determine the power of the weighted sum scalarization with respect to the computation of approximations for general multiobjective minimization and maximization problems. Additionally, we introduce a new multi-factor notion of approximation that is specifically tailored to the multiobjective case and its inherent trade-offs between different objectives.

\smallskip

For minimization problems, we provide an efficient algorithm that computes an approximation of a multiobjective problem by using an exact or approximate algorithm for its weighted sum scalarization. In case that an exact algorithm for the weighted sum scalarization is used, this algorithm comes arbitrarily close to the best approximation quality that is obtainable by supported solutions -- both with respect to the common notion of approximation and with respect to the new multi-factor notion. Moreover, the algorithm yields the currently best approximation results for several well-known multiobjective minimization problems.
For maximization problems, however, we show that a polynomial approximation guarantee can, in general, not be obtained in more than one of the objective functions simultaneously by supported solutions.
\end{abstract}



\section{Introduction}

Almost any real-world optimization problem asks for optimizing more than one objective function (e.g., the minimization of cost and time in transportation systems or the maximization of profit and safety in investments). Clearly, these objectives are conflicting, often incommensurable, and, yet, they have to be taken into account simultaneously. The discipline dealing with such problems is called \emph{multiobjective optimization}. Typically, multiobjective optimization problems are solved according to the Pareto principle of optimality: a solution is called \emph{efficient} (or \emph{Pareto optimal}) if no other feasible solution exists that is not worse in any objective function and better in at least one objective. The images of the efficient solutions in the objective space are called \emph{nondominated points}. In contrast to single objective optimization, where one typically asks for one optimal solution, the main goal of multiobjective optimization is to compute the set of all nondominated points and, for each of them, one corresponding efficient solution. Each of these solutions corresponds to a different compromise among the set of objectives and may potentially be relevant for a decision maker.

\smallskip

Several results in the literature, however, show that multiobjective optimization problems are hard to solve exactly~\cite{Ehrgott:book,Ehrgott:hard-to-say} and, in addition, the cardinalities of the set of nondominated points (the \emph{nondominated set}) and the set of efficient solutions (the \emph{efficient set}) may be exponentially large for discrete problems (and are typically infinite for continuous problems). This impairs the applicability of exact solution methods to real-life problems and provides a strong motivation for studying \emph{approximations of multiobjective optimization problems}.

\smallskip

Both exact and approximate solution methods for multiobjective optimization problems often resort to using single objective auxiliary problems, which are called \emph{scalarizations} of the original multiobjective problem. This refers to the transformation of a multiobjective optimization problem into a single objective auxiliary problem based on a procedure that might use additional parameters, auxiliary points, or variables. The resulting scalarized optimization problems are then solved using methods from single objective optimization and the obtained solutions are interpreted in the context of Pareto optimality.

% One very common approach to tackle multiobjective optimization problems in practice is to assign a positive weight to each objective function and optimize the resulting weighted sum in order to obtain a single solution. If the weights are suitably chosen, this solution might represent a good compromise between the different 

\smallskip

The simplest and most widely used scalarization technique is the \emph{weighted sum scalarization} (see, e.g.,~\cite{Ehrgott:book}). Here, the scalarized auxiliary problem is constructed by assigning a weight to each of the objective functions and summing up the resulting weighted objective functions in order to obtain the objective function of the scalarized problem. If the weights are chosen to be positive, then every optimal solution of the resulting \emph{weighted sum problem} is efficient. Moreover, the weighted sum scalarization does not change the feasible set and, in many cases, boils down to the single objective version of the given multiobjective problem --- which represents an important advantage of this scalarization especially for combinatorial problems. However, only some efficient solutions (called \emph{supported solutions}) can be obtained by means of the weighted sum scalarization, while many other efficient solutions (called \emph{unsupported solutions}) cannot. Consequently, a natural question is to determine which approximations of the whole efficient set can be obtained by using this very important scalarization technique.

\smallskip

\subsection{Previous work}\enlargethispage{\baselineskip} 

Besides many specialized approximation algorithms for particular multiobjective optimization problems, there exist several general approximation methods that can be applied to broad classes of multiobjective problems. An extensive survey of these general approximation methods is provided in~\cite{Herzel+etal:survey}.

Most of these general approximation methods for multiobjective problems are based on the seminal work of Papadimitriou and Yannakakis~\cite{Papadimitriou+Yannakakis:multicrit-approx}, who present a method for generating a $(1+\varepsilon,\dots,1+\varepsilon)$-approximation for general multiobjective minimization and maximization problems with a constant number of positive-valued, polynomially computable objective functions. They show that a $(1+\varepsilon,\dots,1+\varepsilon)$-approximation with size polynomial in the encoding length of the input and $\frac{1}{\varepsilon}$ always exists. Moreover, their results show that the construction of such an approximation is possible in (fully) polynomial time, i.e., the problem admits a \emph{multiobjective (fully) polynomial-time approximation scheme} or \emph{MPTAS} (\emph{MFPTAS}), if and only if a certain auxiliary problem called the \emph{gap problem} can be solved in (fully) polynomial time.
% \begin{prob}[Gap Problem]
% 	Given an instance of a $p$-objective minimization problem, a vector~$b\in \mathbb{R}^p$, and $\varepsilon>0$, either return a feasible solution~$x$ whose objective value~$f(x)$ satisfies $f_j(x)\leq b_j$ for all~$j$ or answer correctly that there is no feasible solution~$x'$ with $f_j(x')\leq \frac{b_j}{1+\varepsilon}$ for all~$j$.
% \end{prob}
% An algorithm that computes a $(1+\varepsilon,\dots,1+\varepsilon)$-approximation in (fully) polynomial time is called a \emph{multiobjective (fully) polynomial-time approximation scheme} or \emph{MPTAS} (\emph{MFPTAS}).
% Hence, the results of Papadimitriou and Yannakakis~\cite{Papadimitriou+Yannakakis:multicrit-approx} show that a multiobjective optimization problem admits an MPTAS (MFPTAS) if and only if the corresponding gap problem is solvable in (fully) polynomial time.

More recent articles building upon the results of~\cite{Papadimitriou+Yannakakis:multicrit-approx} present methods that additionally yield bounds on the size of the computed $(1+\varepsilon,\dots,1+\varepsilon)$-approximation relative to the size of the smallest $(1+\varepsilon,\dots,1+\varepsilon)$-approximation possible~\cite{Vassilvitskii+Yannakakis:trade-off-curves,Diakonikolas+Yannakakis:approx-pareto-sets,Bazgan+etal:min-pareto}. Moreover, it has recently been shown in~\cite{Bazgan+etal:one-exact} that an even better $(1,1+\varepsilon,\dots,1+\varepsilon)$-approximation (i.e., an approximation that is exact in one objective function and $(1+\varepsilon)$-approximate in all other objective functions) always exists, and that such an approximation can be computed in (fully) polynomial time if and only if the so-called \emph{dual restrict problem} (introduced in~\cite{Diakonikolas+Yannakakis:approx-pareto-sets}) can be solved in (fully) polynomial time.

\smallskip

Other works study how the weighted sum scalarization can be used in order to compute a set of solutions such that the convex hull of their images in the objective space yields an approximation guarantee of $(1+\varepsilon,\dots,1+\varepsilon)$~\cite{Diakonikolas+Yannakakis:SODA08,Diakonikolas:Phd,Diakonikolas+Yannakakis:chord-algorithm}. Using techniques similar to ours, Diakonikolas and Yannakakis~\cite{Diakonikolas+Yannakakis:SODA08} show that such a so-called \emph{$\varepsilon$-convex Pareto set} can be computed in (fully) polynomial time if and only if the weighted sum scalarization admits a (fully) polynomial-time approximation scheme. Additionally, they consider questions regarding the cardinality of $\varepsilon$-convex Pareto sets.

\smallskip

Besides the general approximation methods mentioned above that work for both minimization and maximization problems, there exist several general approximation methods that are restricted either to minimization problems or to maximization problems.

For \emph{minimization} problems, there are two general approximation methods that are both based on using (approximations of) the weighted sum scalarization.
The previously best general approximation method for multiobjective minimization problems with an arbitrary constant number of objectives that uses the weighted sum scalarization can be obtained by combining two results of Gla\ss er et al.~\cite{Glasser+etal:multi-hardness,Glasser+etal:CiE2010}. They introduce another auxiliary problem called the \emph{approximate domination problem}, which is similar to the gap problem.
% \begin{prob}[Approximate Domination Problem]
% 	Given an instance of a $p$-objective minimization problem, a vector~$b\in \mathbb{R}^p$, and $\alpha\geq 1$, either return a feasible solution~$x$ whose objective value~$f(x)$ satisfies $f_j(x)\leq \alpha\cdot b_j$ for all~$j$ or answer that there is no feasible solution~$x'$ with $f_j(x')\leq b_j$ for all~$j$.
% \end{prob}
Gla\ss er et al. show that, if this problem is solvable in polynomial time for some approximation factor $\alpha\geq 1$, then an approximating set providing an approximation factor of $\alpha\cdot (1+\varepsilon)$ in every objective function can be computed in fully polynomial time for every $\varepsilon>0$. Moreover, they show that the approximate domination problem with $\alpha\colonequals\sigma\cdot p$ can be solved by using a $\sigma$-approximation algorithm for the weighted sum scalarization of the $p$-objective problem. Together, this implies that a $((1+\varepsilon)\cdot\sigma\cdot p,\ldots,(1+\varepsilon)\cdot\sigma\cdot p)$-approximation can be computed in fully polynomial time for $p$-objective minimization problems provided that the objective functions are positive-valued and polynomially computable and a $\sigma$-approximation algorithm for the weighted sum scalarization exists. %\footnote{Alternatively, applying the algorithm with $\varepsilon'\colonequals \frac{\varepsilon}{\sigma\cdot p}$ yields a $(\sigma\cdot p+\varepsilon,\ldots,\sigma\cdot p+\varepsilon)$-approximation.}
As this result is not explicitly stated in~\cite{Glasser+etal:multi-hardness,Glasser+etal:CiE2010}, no bounds on the running time are provided. 

For \emph{biobjective} minimization problems, Halffmann et al.~\cite{Halffmann+etal:bicriteria} show how to obtain a $(\sigma\cdot(1+2\varepsilon),\sigma\cdot(1+\frac{2}{\varepsilon}))$-approximation for any given $0<\varepsilon\leq 1$ if a polynomial-time $\sigma$-approximation algorithm for the weighted sum scalarization is given.

\smallskip

Obtaining general approximation methods for multiobjective \emph{maximization} problems using the weighted sum scalarization seems to be much harder than for minimization problems. Indeed, Gla\ss er et al.~\cite{Glasser+etal:CiE2010} show that certain translations of approximability results from the weighted sum scalarization of an optimization problem to the multiobjective version that work for minimization problems are not possible in general for maximization problems. %Halffmann et al.~\cite{Halffmann+etal:bicriteria} also show that approximation results similar to the ones they obtain for biobjective minimization problems are impossible to obtain in polynomial time for maximization problems unless $\textsf{P}=\textsf{NP}$.

\smallskip

An approximation method specifically designed for multiobjective maximization problems is presented by Bazgan et al.~\cite{Bazgan+etal:fixed-number}. Their method is applicable to biobjective maximization problems that satisfy an additional structural assumption on the set of feasible solutions and the objective functions: For each two feasible solutions none of which approximates the other one by a factor of~$\alpha$ in both objective functions, a third solution approximating both given solutions in both objective functions by a certain factor depending on $\alpha$ and a parameter~$c$ must be computable in polynomial time. The approximation factor obtained by the algorithm then depends on~$\alpha$ and~$c$. 

\subsection{Our contribution}\label{subsec:our-contribution}

Our contribution is twofold: First, in order to better capture the approximation quality in the context of multiobjective optimization problems, we introduce a new notion of approximation for the multiobjective case. This new notion comprises the common notion of approximation, but is specifically tailored to the multiobjective case and its inherent trade-offs between different objectives. Second, we provide a precise analysis of the approximation quality obtainable for multiobjective optimization problems by means of an exact or approximate algorithm for the weighted sum scalarization -- with respect to both the common and the new notion of approximation. 

\smallskip

In order to motivate the new notion of approximation, consider the biobjective case, in which a $(2+\varepsilon,2+\varepsilon)$-approximation can be obtained from the results of Gla\ss er et al.~\cite{Glasser+etal:multi-hardness,Glasser+etal:CiE2010} using an exact algorithm for the weighted sum scalarization. As illustrated in Figure~\ref{fig:multi-factor-motivation}, this approximation guarantee is actually too pessimistic: Since each point~$y$ in the image of the approximating set is nondominated (since it is the image of an optimal solution of the weighted sum scalarization), no images of feasible solutions can be contained in the shaded region. Thus, every feasible solution is actually either $(1,2+\varepsilon)$- or $(2+\varepsilon,1)$-approximated. Consequently, the approximation quality obtained in this case can be more accurately described by using \emph{two vectors of approximation factors}. In order to capture such situations and allow for a more precise analysis of the approximation quality obtained for multiobjective problems, our new \emph{multi-factor notion of approximation} uses a \emph{set of vectors of approximation factors} instead of only a single vector.

\begin{figure}[ht!]
\pgfdeclarepatternformonly{new north west lines}{%
\pgfqpoint{-1pt}{-1pt}}{\pgfqpoint{4pt}{4pt}}{\pgfqpoint{3pt}{3pt}}%
{
  \pgfsetlinewidth{0.4pt}
  \pgfpathmoveto{\pgfqpoint{0pt}{3pt}}
  \pgfpathlineto{\pgfqpoint{3.1pt}{-0.1pt}}
  \pgfusepath{stroke}
}
	\begin{center}
		\begin{tikzpicture}[scale=1.25]
% 		\fill[gray!30] (1,6.5) -- (0,6.5) -- (0,0) -- (1,0) -- (1,6.5);
 		\fill[gray!30]  (3.25,3.25) -- (1.625,3.25) -- (1.625,1.625) --  (3.25,1.625)  -- (3.25,3.25);
        \fill[line space=7pt, pattern=my north west lines] (7.4,7.4) --  (1.625,7.4) -- (1.625,1.625) -- (7.4,1.625) --  (7.4,7.4);
        \draw[-] (1.625,7.4) -- (1.625,1.625) -- (7.4,1.625);
        \draw[-] (1.625,3.25) -- (7.4,3.25);
        \draw[-] (3.25,1.625) -- (3.25,7.4);
		\draw[->] (-0.2,0) -- (7.4,0) node[below right] {$f_1$};
		\draw[->] (0,-0.2) -- (0,7.4) node[above left] {$f_2$};
		\fill (3.25,3.25) circle (1.5pt) node[above right] {$y$};
		\end{tikzpicture}
		\caption{Image space of a biobjective minimization problem. The point~$y$ in the image of the approximating set $(2+\varepsilon,2+\varepsilon)$-approximates all points in the hashed region. If~$y$ is nondominated, no images of feasible solutions can be contained in the shaded region, so every image in the hashed region is actually either $(1,2+\varepsilon)$- or $(2+\varepsilon,1)$-approximated.}\label{fig:multi-factor-motivation}
	\end{center}
\end{figure}
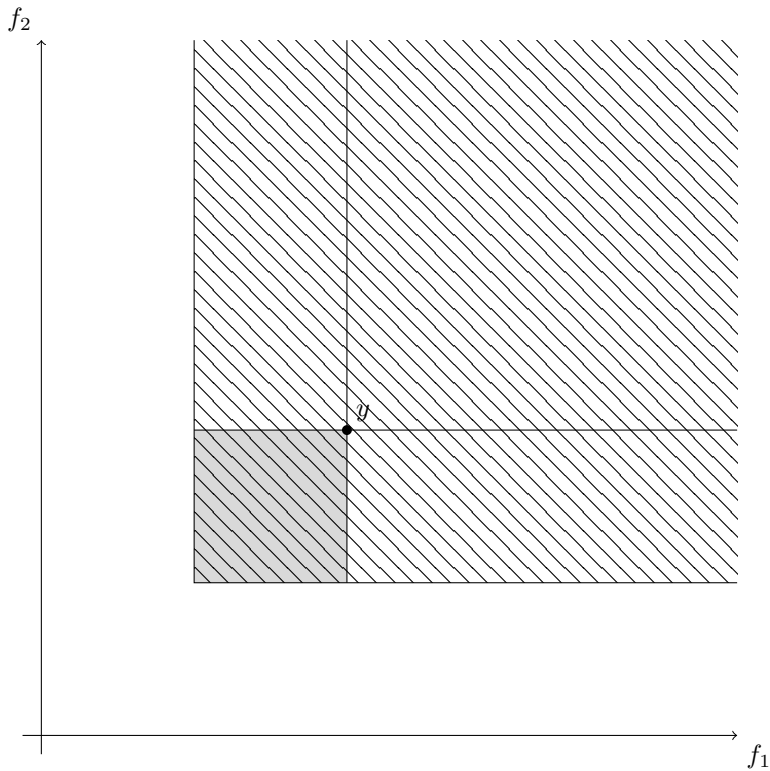

\newpage

The second part of our contribution consists of a detailed analysis of the approximation quality obtainable by using the weighted sum scalarization -- both for multiobjective minimization problems and for multiobjective maximization problems.
For minimization problems, we provide an efficient algorithm that approximates a multiobjective problem using an exact or approximate algorithm for its weighted sum scalarization. We analyze the approximation quality obtained by the algorithm both with respect to the common notion of approximation that uses only a single vector of approximation factors as well as with respect to the new multi-factor notion. With respect to the common notion, our algorithm matches the best previously known approximation guarantee of $(\sigma\cdot p+\varepsilon,\ldots,\sigma\cdot p+\varepsilon)$ obtainable for $p$-objective minimization problems and any $\varepsilon>0$ from a $\sigma$-approximation algorithm for the weighted sum scalarization. More importantly, we show that this result is best-possible in the sense that it comes arbitrarily close to the best approximation guarantee obtainable by supported solutions for the case that an exact algorithm is used to solve the weighted sum problem (i.e., when $\sigma=1$).

\smallskip

When analyzing the algorithm with respect to the new multi-factor notion of approximation, however, a much stronger approximation result is obtained. Here, we show that every feasible solution is approximated with some (possibly different) vector $(\alpha_1,\dots,\alpha_p)$ of approximations factors such that $\sum_{j:\alpha_j>1}\alpha_j = \sigma\cdot p + \varepsilon$. In particular, the worst-case approximation factor of $\sigma\cdot p + \varepsilon$ can actually be tight \emph{in at most one objective} for any feasible point. This shows the multi-factor notion of approximation yields a much stronger approximation result by allowing a refined analysis of the obtained approximation guarantee. Moreover, for $\sigma=1$, we show that the obtained multi-factor approximation result comes arbitrarily close to the best multi-factor approximation result obtainable by supported solutions. We also demonstrate that our algorithm applies to a large variety of multiobjective minimization problems and yields the currently best approximation results for several problems.

\smallskip

Multiobjective maximization problems, however, turn out to be much harder to approximate by using the weighted sum scalarization. Here, we show that a polynomial approximation guarantee can, in general, not be obtained in more than one of the objective functions simultaneously when using only supported solutions.

\smallskip

In summary, our results yield essentially tight bounds on the power of the weighted sum scalarization with respect to the approximation of multiobjective minimization and maximization problems -- both in the common notion of approximation and in the new multi-factor notion.

\medskip

The remainder of the paper is organized as follows: In Section~\ref{sec:preliminaries}, we formally introduce multiobjective optimization problems and provide the necessary definitions concerning their approximation. Section~\ref{sec:minimization} contains our general approximation algorithm for minimization problems (Subsection~\ref{subsec:results}) as well as a faster algorithm for the biobjective case (Subsection~\ref{subsec:biobjective}). Moreover, we show in Subsection~\ref{subsec:tightness} that the obtained approximation results are tight. Section~\ref{sec:applications} presents applications of our results to specific minimization problems. In Section~\ref{sec:maximixation}, we present our impossibility results for maximization problems. Section~\ref{sec:conclusion} concludes the paper and lists directions for future work.

\newpage

\section{Preliminaries}\label{sec:preliminaries}
In the following, we consider a general multiobjective minimization or maximization problem~$\Pi$ of the following form (where either all objective functions are to be minimized or all objective functions are to be to maximized):
\begin{align*}
\min/\max\; & f(x)=(f_1(x),\dots,f_p(x))\\ 
\text{s.\,t. } & x \in X
\end{align*}
Here, as usual, we assume a constant number~$p\geq 2$ of objectives.
The elements $x\in X$ are called \emph{feasible solutions} and the set~$X$, which is assumed to be nonempty, is referred to as the \emph{feasible set}. An image $y=f(x)$ of a feasible solution $x\in X$ is also called a \emph{feasible point}. We let $Y\colonequals f(X) \colonequals \{f(x): x\in X\}\subseteq \mathbb{R}^p$ denote the \emph{set of feasible points}.

We assume that the objective functions take only positive rational values and are polynomially computable. Moreover, for each $j\in\{1,\dots,p\}$, we assume that there exist strictly positive rational lower and upper bounds $\LB(j),\UB(j)$ of polynomial encoding length such that $\LB(j) \leq f_j(x) \leq \UB(j)$ for all $x\in X$. We let $\LB\colonequals \min_{j=1,\dots,p}\LB(j)$ and $\UB\colonequals \max_{j=1,\dots,p}\UB(j)$.

\medskip

\begin{definition}
For a minimization problem~$\Pi$, we say that a point $y=f(x)\in Y$ is \emph{dominated} by another point $y'=f(x')\in Y$ if $y'\neq y$ and
\begin{align*}
y'_j=f_j(x') \leq f_j(x)=y_j \text{ for all } j\in\{1,\dots,p\}.
\end{align*}

Similarly, for a maximization problem~$\Pi$, we say that a point $y=f(x)\in Y$ is \emph{dominated} by another point $y'=f(x')\in Y$ if $y'\neq y$ and
\begin{align*}
y'_j=f_j(x') \geq  f_j(x)=y_j \text{ for all } j\in\{1,\dots,p\}.
\end{align*}
If the point $y=f(x)$ is not dominated by any other point~$y'$, we call~$y$ \emph{nondominated} and the feasible solution $x\in X$ \emph{efficient}. The set~$\YN$ of nondominated points is called the \emph{nondominated set} and the set~$\XE$ of efficient solutions is called the \emph{efficient set} or \emph{Pareto set}.
\end{definition}

%Note that the notion of the Pareto set is somewhat blurred in the literature as it is sometimes used to describe the set of efficient solutions and sometimes to describe the set of nondominated images. We exclusively use the term Pareto set to refer to the set of efficient solutions of a problem here.

\subsection{Notions of approximation}

We first recall the standard definitions of approximation for single objective optimization problems.

\begin{definition}
Consider a single objective optimization problem~$\Pi$ and let $\alpha\geq 1$. If $\Pi$ is a minimization problem, we say that a feasible solution $x\in X$ \emph{$\alpha$-appro\-xi\-mates} another feasible solution $x'\in X$ if $f(x) \leq \alpha \cdot f(x')$. If $\Pi$ is a maximization problem, we say that a feasible solution $x\in X$ \emph{$\alpha$-approximates} another feasible solution $x'\in X$ if $\alpha\cdot f(x) \geq  f(x')$.
A feasible solution that $\alpha$-approximates every feasible solution of~$\Pi$ is called \emph{an $\alpha$-approximation} for~$\Pi$.
% A feasible solution that $\alpha$-approximates an optimal solution of~$\Pi$ is called \emph{an $\alpha$-approximation} for~$\Pi$.

A \emph{(polynomial-time) $\alpha$-approximation algorithm} is an algorithm that, for every instance~$I$ with encoding length~$|I|$, computes an $\alpha$-approximation for~$\Pi$ in time bounded by a polynomial in~$|I|$.
\end{definition}

\noindent
The following definition extends the concept of approximation to the multiobjective case.

\begin{definition}
Let $\alpha=(\alpha_1,\dots,\alpha_p)\in\mathbb{R}^p$ with $\alpha_j\geq 1$ for all $j\in\{1,\dots,p\}$.

For a minimization problem~$\Pi$, we say that a feasible solution $x\in X$ \emph{$\alpha$-appro\-xi\-mates} another feasible solution $x'\in X$ (or, equivalently, that the feasible point $y=f(x)$ \emph{$\alpha$-appro\-xi\-mates} the feasible point $y'=f(x')$) if
\begin{align*}
f_j(x) \leq \alpha_j\cdot f_j(x') \text{ for all } j\in\{1,\dots,p\}.
\end{align*}

Similarly, for a maximization problem~$\Pi$, we say that a feasible solution $x\in X$ \emph{$\alpha$-approximates} another feasible solution $x'\in X$ (or, equivalently, that the feasible point $y=f(x)$ \emph{$\alpha$-appro\-xi\-mates} the feasible point $y'=f(x')$) if
\begin{align*}
\alpha_j\cdot f_j(x) \geq  f_j(x') \text{ for all } j\in\{1,\dots,p\}.
\end{align*}
\end{definition}

The standard notion of approximation for multiobjective optimization problems   used in the literature is the following one.

\begin{definition}\label{def:approx-standard}
Let $\alpha=(\alpha_1,\dots,\alpha_p)\in\mathbb{R}^p$ with $\alpha_j\geq 1$ for all $j\in\{1,\dots,p\}$.

A set~$P\subseteq X$ of feasible solutions is called an \emph{$\alpha$-approximation} for the multiobjective problem~$\Pi$ if, for any feasible solution~$x'\in X$, there exists a solution~$x\in P$ that $\alpha$-approximates~$x'$.
\end{definition}

In the following definition, we generalize the standard notion of approximation for multiobjective problems by allowing a \emph{set of vectors of approximation factors} instead of only a single vector, which allows for tighter approximation results.

\begin{definition}\label{def:approx-generalized}
Let $\mathcal{A}\subseteq\mathbb{R}^p$ be a set of vectors with $\alpha_j\geq 1$ for all $\alpha\in \mathcal{A}$ and all $j\in\{1,\dots,p\}$. Then a set~$P\subseteq X$ of feasible solutions is called a \emph{(multi-factor) $\mathcal{A}$-approximation} for the multiobjective problem~$\Pi$ if, for any feasible solution~$x'\in X$, there exists a solution~$x\in P$ and a vector~$\alpha\in\mathcal{A}$ such that~$x$ $\alpha$-approximates~$x'$.
\end{definition}

% \noindent
% In the following definition, we generalize the standard notion of approximation for multiobjective problems in two ways: First, we allow a \emph{set of vectors of approximation factors} instead of only a single vector~$\alpha$ as in Definition~\ref{def:approx-standard}, which will allow for tighter approximation results. Second, we consider approximations of arbitrary subsets of the feasible set.

% \begin{definition}\label{def:approx-generalized}
% Let $X'\subseteq X$ be any subset of the feasible set and let $\mathcal{A}\subseteq\mathbb{R}^p$ be a set of vectors with $\alpha_j\geq 1$ for all $\alpha\in \mathcal{A}$ and all $j\in\{1,\dots,p\}$. Then a set~$P\subseteq X$ of feasible solutions is called a \emph{(multi-factor) $\mathcal{A}$-approximation of~$X'$} if, for any feasible solution~$x'\in X'$, there exists a solution~$x\in P$ and a vector~$\alpha\in\mathcal{A}$ such that~$x$ $\alpha$-approximates~$x'$.

% An $\mathcal{A}$-approximation of the whole feasible set~$X$ is also called an $\mathcal{A}$-approximation for the multiobjective problem~$\Pi$.
% \end{definition}

Note that, in the case where $\mathcal{A}=\{(\alpha_1,\dots,\alpha_p)\}$ is a singleton, an $\mathcal{A}$-appro\-xi\-ma\-tion for a multiobjective problem according to Definition~\ref{def:approx-generalized} is equivalent to an $(\alpha_1,\dots,\alpha_p)$-approximation according to Definition~\ref{def:approx-standard}.

%Moreover, the following properties (which we state here for minimization problems) follow directly from the definition:

% \begin{obs}\label{obs:multifactor-notion}
% Let~$P\subseteq X$ be an $\mathcal{A}$-approximation for the $p$-objective minimization problem~$\Pi$.
% \begin{enumerate}
% \item[(a)] If $\alpha,\beta\in\mathcal{A}$ with $\alpha\leq\beta$, then~$P$ is also an $(\mathcal{A}\setminus\{\alpha\})$-approximation for~$\Pi$.
% \item[(b)] If $\mathcal{A}'\subseteq\mathbb{R}^p$ is such that, for each $\alpha\in\mathcal{A}$, there exists $\alpha'\in\mathcal{A}'$ with $\alpha'\geqq \alpha$, then~$P$ is also an $\mathcal{A}'$-approximation for~$\Pi$.
% \end{enumerate}
% \end{obs}

% Note Observation~\ref{obs:multifactor-notion} is stated for minimization problems here, but the vectors~$\alpha$ that can be removed from the set~$\mathcal{A}$ of approximation factors by part~(a) are the ones which are dominated in the \emph{maximization} sense.

% Also note that Observation~\ref{obs:multifactor-notion}~(a) does \emph{not} imply that one can always replace the set~$\mathcal{A}$ by its nondominated part~$\mathcal{A}_N$ (in the maximization sense) since  there might exist dominated vectors in~$\mathcal{A}$ that are not dominated by any vector from~$\mathcal{A}_N$ (the set~$\mathcal{A}_N$ might even be empty although $\mathcal{A}$ is nonempty -- see, e.g.,~\cite{Ehrgott:book}). 

\subsection{Weighted sum scalarization}
Given a $p$-objective optimization problem~$\Pi$ and a vector $w=(w_1,\ldots,w_p)\in\mathbb{R}^p$ with $w_j>0$ for all $j \in\{1,\ldots,p\}$,  the \emph{weighted sum problem} (or \emph{weighted sum scalarization})~$\Pi^{\WS}(w)$ associated with~$\Pi$ is defined as the following single objective optimization problem: 
\begin{align*}
\min/\max\; & \sum_{j=1}^{p} w_j \cdot f_{j}(x)\\
\text{s.\,t. } & x \in X
\end{align*}

\begin{definition}
A feasible solution~$x\in X$ is called \emph{supported} if there exists a vector $w=(w_1,\ldots,w_p) \in\mathbb{R}^p$ of positive weights such that~$x$ is an optimal solution of the weighted sum problem~$\Pi^{\WS}(w)$. In this case, the feasible point $y=f(x)\in Y$ is called a \emph{supported point}. The set of all supported solutions is denoted by~$\XS$.
% A point $y=f(x)\in Y$ is called \emph{supported} if there exists a vector $w=(w_1,\ldots,w_p) \in\mathbb{R}^p$ of positive weights such that~$y$ is an optimal solution of the weighted sum problem~$\Pi^{\WS}(w)$. In this case, the feasible solution $x\in X$ is called a \emph{supported solution}. The set of all supported solutions will be denoted by~$\XS$.
% If the supported point $y=f(x)\in Y$ is an extreme point of~$Y$, then~$y$ is called an \emph{extreme supported point} and~$x\in X$ is called an \emph{extreme supported solution}. The set of all extreme supported solutions will be denoted by~$\XXS$.
\end{definition}

It is well-known that every supported point is nondominated and, correspondingly, every supported solution is efficient (cf.~\cite{Ehrgott:book}).

\smallskip

In the following, we assume that there exists a polynomial-time $\sigma$-approxi\-ma\-tion algorithm $\WS_{\sigma}$ for the weighted sum problem, where $\sigma\geq 1$ can be either a constant or a function of the input size. When calling $\WS_{\sigma}$ with some specific weight vector~$w$, we denote this by $\WS_{\sigma}(w)$. This algorithm then returns a solution~$\hat x$ such that
$\sum_{j=1}^{p} w_j f_j(\hat x) \leq \sigma \cdot \sum_{j=1}^{p} w_j f_j(x)$ for all $x\in X$, if $\Pi$ is a minimization problem, and $\sigma \cdot \sum_{j=1}^{p} w_j f_j(\hat x) \geq  \sum_{j=1}^{p} w_j f_j(x)$ for all $x\in X$, if $\Pi$ is a maximization problem. The running time of algorithm $\WS_{\sigma}$ is denoted by $T_{WS_{\sigma}}$.

% In the following, we assume that there exists a polynomial-time $\sigma$-approxi\-ma\-tion algorithm $\WS_{\sigma}$ for the weighted sum problem, where $\sigma\geq 1$ can be either a constant or a function of the input size. When calling $\WS_{\sigma}$ with some specific weight vector~$w$, we denote this by $\WS_{\sigma}(w)$. This algorithm then returns a solution $\hat x$ such that
% $ \sum_{j=1}^{p} w_j f_j(\hat x) \leq \sigma \cdot \sum_{j=1}^{p} w_j f_j(x^*)$, if $\Pi$ is a minimization problem, and  $ \sigma \cdot \sum_{j=1}^{p} w_j f_j(\hat x) \geq  \sum_{j=1}^{p} w_j f_j(x^*)$, if $\Pi$ is a maximization problem, where $x^*$ is an optimal solution of $\Pi^{\WS}(w)$. The running time of algorithm $\WS_{\sigma}$ is denoted by $T_{WS_{\sigma}}$.

\medskip

The following result shows that a $\sigma$-approximation for the weighted sum problem is also a $\sigma$-approximation of any solution in at least one of the objectives.

\begin{lemma}\label{lem:sigma_efficiency}
Let~$\hat{x} \in X$ be a $\sigma$-approximation for~$\Pi^{\WS}(w)$ for some positive weight vector~$w\in\mathbb{R}^p$. Then, for any feasible solution~$x \in X$, there exists at least one $i \in \{1,\ldots,p\}$ such that $\hat{x}$ $\sigma$-approximates~$x$ in objective~$f_i$. 
\end{lemma}

\begin{proof}
Consider the case where~$\Pi$ is a multiobjective \emph{minimization} problem (the proof for the case where~$\Pi$ is a maximization problem works analogously). Then, we must show that, for any feasible solution $x \in X$, there exists at least one $i \in \{1,\ldots,p\}$ such that $f_i(\hat{x}) \leq \sigma \cdot f_i(x)$.

Assume by contradiction that there exists some~$x' \in X$ such that $f_j(\hat{x}) > \sigma \cdot f_j(x')$ for all $j \in \{1,\ldots,p\}$. Then, we obtain $\sum_{j=1}^{p} w_j \cdot f_{j}(\hat{x}) > \sigma \cdot  \sum_{j=1}^{p} w_j \cdot f_{j}(x')$, which contradicts the assumption that~$\hat{x}$ is a $\sigma$-approximation for $\Pi^{\WS}(w)$.
\end{proof}

% \newpage

%\section{Approximating the nondominated set of multi-objective minimization problems}
\section{A multi-factor approximation result for minimization problems}\label{sec:minimization}

In this section, we study the approximation of multiobjective \emph{minimization} problems by (approximately) solving weighted sum problems. In Subsection~\ref{subsec:results}, we propose a multi-factor approximation algorithm that significantly improves upon the $((1+\varepsilon)\cdot\sigma\cdot p, \dots,(1+\varepsilon)\cdot\sigma\cdot p)$-approximation algorithm that can be derived from Gla\ss er et al.~\cite{Glasser+etal:CiE2010}. The biobjective case is then investigated in Subsection~\ref{subsec:biobjective}. Finally, we show in Subsection~\ref{subsec:tightness} that the resulting approximation is tight.

\subsection{General results}\label{subsec:results}

\begin{proposition}\label{prop_1WS}
Let $\bar{x}\in X$ be a feasible solution of~$\Pi$ and let $b=(b_1,\dots,b_p)$ be such that  $b_j\leq f_j(\bar{x})\leq (1+\varepsilon)\cdot b_j$ for $j=1,\dots,p$ and some $\varepsilon>0$.
Applying~$\WS_{\sigma}(w)$ with $w_j\colonequals\frac{1}{b_j}$ for $j=1,\dots,p$ yields a solution~$\hat x$ that $(\alpha_1,\dots,\alpha_p)$-approximates~$\bar x$ for some $\alpha_1,\dots,\alpha_p\geq1$ such that $\alpha_i \leq \sigma$ for at least one $i \in \{1,\ldots,p\}$ and
\begin{align*}
	\sum_{j:\alpha_j>1} \alpha_j = (1+\varepsilon)\cdot\sigma\cdot p.
\end{align*}
\end{proposition}

\begin{proof}
Since~$\hat x$ is the solution returned by $\WS_{\sigma}(w)$, we have
\begin{align*}
  \sum_{j=1}^p\frac{1}{b_j} f_j(\hat{x})
  \leq \sigma\cdot\left(\sum_{j=1}^p\frac{1}{b_j}f_j(\bar{x})\right)
  \leq \sigma\cdot(1+\varepsilon)\cdot\left(\sum_{j=1}^p 1\right) = (1+\varepsilon)\cdot\sigma\cdot p.
\end{align*}

% Let~$x^*$ be an optimal solution for $\Pi^{\WS}(w)$. Since~$\hat x$ is the solution returned by $\WS_{\sigma}(w)$, we have
% \begin{align*}
% \sum_{j=1}^p\frac{1}{b_j} f_j(\hat{x})
%  & \leq \sigma\cdot\left(\sum_{j=1}^p\frac{1}{b_j}f_j(x^*)\right)
%     \leq \sigma\cdot\left(\sum_{j=1}^p\frac{1}{b_j}f_j(\bar{x})\right) \\
%  & \leq \sigma\cdot(1+\varepsilon)\cdot\left(\sum_{j=1}^p 1\right) = (1+\varepsilon)\cdot\sigma\cdot p.
% \end{align*}

\noindent Since $\frac{1}{b_j} \geq \frac{1}{f_j(\bar{x})}$, we get $\sum_{j=1}^p  \frac{f_j(\hat{x})}{f_j(\bar{x})} \leq \sum_{j=1}^p\frac{1}{b_j} f_j(\hat{x})$, which yields

\begin{align*}
	\sum_{j=1}^p  \frac{f_j(\hat{x})}{f_j(\bar{x})} \leq (1+\varepsilon)\cdot\sigma\cdot p.
\end{align*}

\noindent  Setting $\alpha_j \colonequals \max\left\{1,\frac{f_j(\hat{x})}{f_j(\bar{x})}\right\}$ for $j=1,\dots,p$, we have
\begin{align*}
	\sum_{j:\alpha_j>1} \alpha_j \leq (1+\varepsilon)\cdot\sigma\cdot p. %\label{eq:guarantee}
\end{align*}
The worst case approximation factors $\alpha_j$ are then obtained when equality holds in the previous inequality.
Moreover, by Lemma~\ref{lem:sigma_efficiency}, there exists at least one $i \in \{1,\ldots,p\}$ such that $f_i(\hat{x}) \leq \sigma \cdot f_i(\bar{x})$. Thus, we have $\alpha_i\leq \sigma$ for at least one $i \in \{1,\ldots,p\}$, which proves the claim.
% which implies that
% \begin{align*}
% 	\sum_{j:\frac{f_j(\hat{x})}{f_j(\bar{x})}>\sigma}  \frac{f_j(\hat{x})}{f_j(\bar{x})}< (1+\varepsilon)\cdot\sigma\cdot p.
% \end{align*}
\end{proof}

%\noindent
% Note that, by choosing $\varepsilon'\colonequals \frac{\varepsilon}{\sigma\cdot p}$, we could replace~\eqref{eq:guarantee} in Proposition~\ref{prop_1WS} by
% \begin{align*}
% 	\sum_{j:\alpha_j>\sigma} \alpha_j < \sigma\cdot p+\varepsilon.
% \end{align*}

Proposition~\ref{prop_1WS} motivates to apply the given $\sigma$-approximation algorithm~$\WS_{\sigma}$ for~$\Pi^{\WS}$ iteratively for different weight vectors~$w$ in order to obtain an approximation of the multiobjective minimization problem~$\Pi$. This is formalized in Algorithm~\ref{alg:mainAlgo}, whose correctness and running time are established in Theorem~\ref{thm:main-result}.

\begin{algorithm2e}
\SetKw{Compute}{compute}
\SetKw{Break}{break}
\SetKwInOut{Input}{input}\SetKwInOut{Output}{output}
\SetKwComment{command}{right mark}{left mark}

\Input{an instance of a $p$-objective minimization problem~$\Pi$, $\varepsilon >0$, a $\sigma$-approxi\-ma\-tion algorithm $\WS_{\sigma}$ for the weighted sum problem}

\Output{an \emph{$\mathcal{A}$-approximation} $P$ for problem~$\Pi$}

\BlankLine
$P \leftarrow \emptyset$

$\varepsilon' \leftarrow \frac{\varepsilon}{\sigma\cdot p}$

\For{$j\leftarrow 1$ \KwTo $p$}{

$u_j \leftarrow$ largest integer such that $LB(j)\cdot(1+\varepsilon')^{u_j} \leq UB(j)$
}

\For{$k\leftarrow 1$ \KwTo $p$}{

$i_k \leftarrow 0$

\ForEach {$(i_1,\ldots,i_p)$ such that $i_{\ell} \in \{1,\ldots,u_{\ell}\}$ for $\ell < k$, and $i_{\ell} \in \{0,\ldots,u_{\ell}\}$ for $\ell > k$}
{

\For{$j\leftarrow 1$ \KwTo $p$}{

$b_j \leftarrow LB(j)\cdot(1 + \varepsilon')^{i_j}$

$w_j \leftarrow \frac{1}{b_j}$
}
$x \leftarrow \WS_{\sigma}(w)$

$P \leftarrow P \cup \{x\}$
}
}
\Return $P$
\caption{An \emph{$\mathcal{A}$-approximation for $p$-objective minimization problems}\label{alg:mainAlgo}}
\end{algorithm2e}

\begin{theorem}\label{thm:main-result}
For a $p$-objective minimization problem, Algorithm~\ref{alg:mainAlgo} outputs an $\mathcal{A}$-approxi\-ma\-tion where
\begin{align*}
%\mathcal{A} = \left\{(1+\varepsilon)\cdot(\alpha_1,\dots,\alpha_p) : \alpha_1,\dots,\alpha_p \geq 1 \mbox{ and }\sum_{j:\alpha_j>1} \alpha_j \leq \sigma\cdot p\right\}}
\mathcal{A} = & \{(\alpha_1,\dots,\alpha_p) : \alpha_1,\dots,\alpha_p \geq 1, \alpha_i \leq \sigma \mbox{ for at least one $i$, and} \sum_{j:\alpha_j>1} \alpha_j = \sigma\cdot p\ + \varepsilon \}
\end{align*}
in time bounded by $\displaystyle T_{WS_{\sigma}}\cdot \sum_{i=1}^p \prod_{j\neq i}\left\lceil \log_{1+\frac{\varepsilon}{\sigma p}}\frac{\UB(j)}{\LB(j)}\right\rceil \in \mathcal{O}\left(T_{WS_{\sigma}} \left(\frac{\sigma}{\varepsilon} \log\frac{\UB}{\LB}\right)^{p-1} \right)$.
\end{theorem}

\begin{proof}
In order to approximate all feasible solutions, we can iteratively apply Proposition~\ref{prop_1WS} with $\varepsilon' \colonequals \frac{\varepsilon}{\sigma\cdot p}$ instead of $\varepsilon$, leading to the modified constraint on the sum of the~$\alpha_j$ where the right-hand side becomes $(1+\varepsilon')\cdot\sigma\cdot p = \sigma\cdot p\ + \varepsilon$. More precisely, we iterate with $b_j = LB(j)\cdot(1+ \varepsilon')^{i_j}$ and $i_j=0,\ldots,u_j$, where $u_j$ is the largest integer such that $LB(j)\cdot(1+\varepsilon')^{u_j} \leq UB(j)$, for each $j\in\{1,\ldots,p\}$. Actually, this iterative application of Proposition~\ref{prop_1WS} involves redundant weight vectors. More precisely, consider a weight vector $w = (w_1,\ldots,w_p)$ where $w_j= \frac{1}{b_j}$ with $b_j = LB(j)\cdot(1+\varepsilon')^{t_j}$ for $j=1,\ldots,p$, and let~$k$ be an index such that $t_k  = \min_{j=1,\ldots,p} t_j$. Then problem~$\Pi^{\WS}(w)$ is equivalent to problem~$\Pi^{\WS}(w')$ with $w'_j =\frac{1}{b'_j}$, where $b'_j = LB(j)\cdot(1+\varepsilon')^{t_j - t_k}$ for $j=1,\ldots,p$. Therefore, it is sufficient to consider all weight vectors~$w$ for which at least one component~$w_k$ is set to $\frac{1}{LB(k)}$ (see Figure~\ref{fig:approx-grid} for an illustration). The running time follows.
%t\cdot LB(j)$ from $t=1$ to the smallest integer $u$ such that $(1+ \varepsilon)^t\cdot LB(j) \geq UB(j)$ for $j=1,\ldots,p$.
\end{proof}

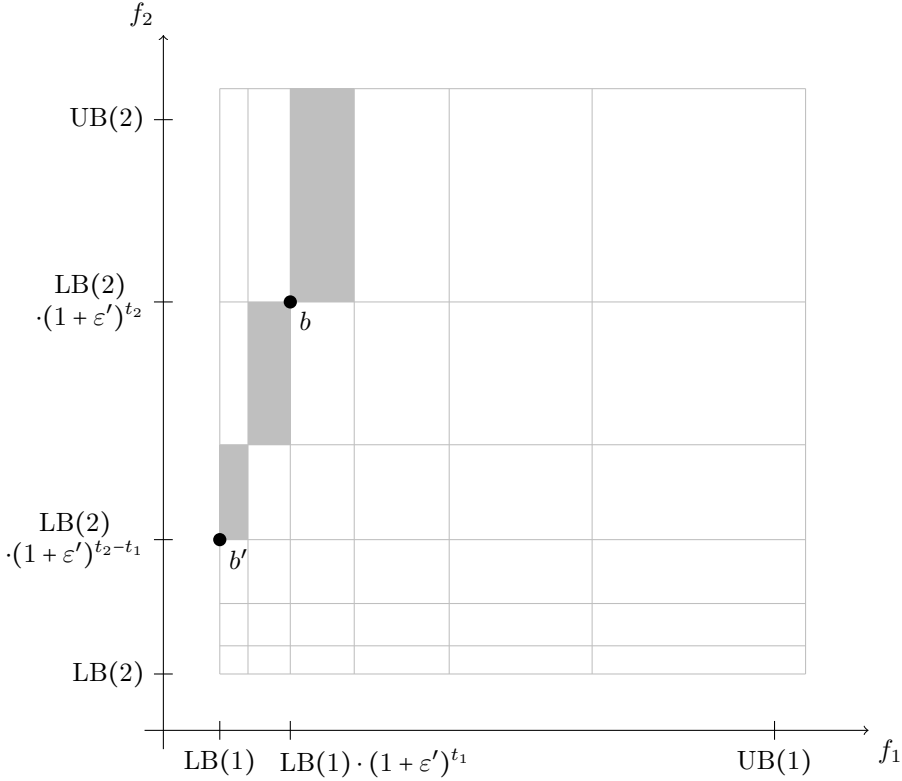
\begin{figure}[ht!]
\begin{center}
	\begin{tikzpicture}[scale=1.25]
	% Axes
	\draw[->] (-0.2,0) -- (7.5,0) node[below right] {$f_1$};
	\draw[->] (0,-0.2) -- (0,7.4) node[above left] {$f_2$};
	% Grid:
	\foreach \x in {0.6, 0.9, 1.35, 2.03, 3.04, 4.56, 6.83}{
		\draw[-,thin,color=gray!50!white] (\x,0.6) -- (\x,6.83) node[below] {};
		\draw[-,thin,color=gray!50!white] (0.6,\x) -- (6.83,\x) node[below] {};
	}
	% Rest
    \fill [gray!50] (0.6,2.03) rectangle (0.9,3.04);
    \fill [gray!50] (0.9,3.04) rectangle (1.35,4.56);
    \fill [gray!50] (1.35,4.56) rectangle (2.03,6.83);
%     \fill (1.9,5.7) circle (2pt) node[above right] {$\; f(\bar{x})$};
    \fill (1.35,4.56) circle (2pt) node[below right] {$b$};
    \fill (0.6,2.03) circle (2pt) node[below right] {$b'$};
	\draw[-] (0.6,0.1) -- (0.6,-0.1) node[below] {$\LB(1)$};
 	\draw[-] (1.35,0.1) -- (1.35,-0.1) node[below right] {\hspace{-0.75em}$\LB(1)\cdot(1+\varepsilon')^{t_1}$};
	\draw[-] (6.5,0.1) -- (6.5,-0.1) node[below] {$\UB(1)$};
	\draw[-] (0.1,0.6) -- (-0.1,0.6) node[left] {$\LB(2)$};
	\draw[-] (0.1,2.03) -- (-0.1,2.03) node[left,align=center] {$\LB(2)$\\$\cdot(1+\varepsilon')^{t_2-t_1}$};
    \draw[-] (0.1,4.56) -- (-0.1,4.56) node[left,align=center] {$\LB(2)$\\$\cdot(1+\varepsilon')^{t_2}$};
	\draw[-] (0.1,6.5) -- (-0.1,6.5) node[left] {$\UB(2)$};
	\end{tikzpicture}
    \caption{Weight vectors and subdivision of the objective space in Algorithm~\ref{alg:mainAlgo}. The weight vector~$w=(\frac{1}{b_1},\dots,\frac{1}{b_p})$ with $b_j = LB(j)\cdot(1+\varepsilon')^{t_j}$ for $j=1,\ldots,p$ is equivalent to the weight vector $w'=(\frac{1}{b'_1},\dots,\frac{1}{b'_p})$ obtained by reducing all exponents~$t_j$ by their minimum. The solution~$\WS_{\sigma}(w')$ returned for~$w'$ is then used to approximate all solutions with images in the shaded (hyper-) rectangles.}\label{fig:approx-grid}
% 	\caption{Weight vectors and subdivision of the objective space in Algorithm~\ref{alg:mainAlgo}. For each weight vector~$w=(\frac{1}{b_1},\dots,\frac{1}{b_p})$ considered in the algorithm (where $b_k=\LB(k)$ for at least one~$k$), all solutions~$\bar{x}$ with images in the hyperrectangles $\bigtimes_{j=1}^p \left[b_j\cdot(1+\varepsilon)^\ell,b_j\cdot(1+\varepsilon)^{\ell+1}\right]$ for $\ell=0,1,\dots$ are approximated by the solution returned by~$\WS_{\sigma}(w)$.}\label{fig:approx-grid}
\end{center}
\end{figure}

Note that, depending on the structure of the weighted sum algorithm~$\WS_{\sigma}$, the practical running time of Algorithm~\ref{alg:mainAlgo} could be improved by not solving every weighted sum problem from scratch, but using the information obtained in previous iterations.

Also note that, as illustrated in Figure~\ref{fig:approx-grid}, Algorithm~\ref{alg:mainAlgo} also directly yields a subdivision of the objective space into hyperrectangles such that all solutions whose images are in the same hyperrectangle are approximated by the same solution (possibly with different approximation guarantees): For each weight vector~$w=(\frac{1}{b_1},\dots,\frac{1}{b_p})$ considered in the algorithm (where $b_k=\LB(k)$ for at least one~$k$), all solutions~$\bar{x}$ with images in the hyperrectangles $\bigtimes_{j=1}^p \left[b_j\cdot(1+\varepsilon')^\ell,b_j\cdot(1+\varepsilon')^{\ell+1}\right]$ for $\ell=0,1,\dots$ are approximated by the solution returned by~$\WS_{\sigma}(w)$.

\medskip

\noindent
When the weighted sum problem can be solved exactly in polynomial time, Theorem~\ref{thm:main-result} immediately yields the following result:

\begin{corollary}\label{cor:main-result-special}
If $\WS_{\sigma}=\WS_{1}$ is an exact algorithm for the weighted sum problem, Algorithm~\ref{alg:mainAlgo}  outputs an $\mathcal{A}$-approximation where
\begin{align*}
\mathcal{A} = & \{(\alpha_1,\dots,\alpha_p) : \alpha_1,\dots,\alpha_p \geq 1, \alpha_i=1 \text{ for at least one}~i, \mbox{ and } \sum_{j:\alpha_j>1} \alpha_j = p + \varepsilon\}
\end{align*}
in time %bounded by $\displaystyle T_{\WS_{1}} \sum_{i=1}^p \prod_{j\neq i}\left\lceil \log_{1+\varepsilon}\frac{\UB(j)}{\LB(j)}\right\rceil \in 
$\mathcal{O}\left(T_{\WS_{1}} \left(\frac{1}{\varepsilon} \log\frac{\UB}{\LB}\right)^{p-1} \right)$.
\end{corollary}

Another special case worth mentioning is the situation where the weighted sum problem admits a polynomial-time approximation scheme. Here, similar to the case in which an exact algorithm is available for the weighted sum problem (see Corollary~\ref{cor:main-result-special}), we can still obtain a set of vectors~$\alpha$ of approximation factors with $\sum_{j:\alpha_j>1}\alpha_j=p+\varepsilon$ while only losing the property that at least one $\alpha_i$ equals~$1$.

\begin{corollary}\label{cor:weighted-sum-PTAS}
If the weighted sum problem admits a polynomial-time $(1+\tau)$-approximation for any $\tau>0$, then, for any $\varepsilon>0$ and any $0<\tau<\frac{\varepsilon}{p}$, Algorithm~\ref{alg:mainAlgo} can be used to compute an $\mathcal{A}$-approximation where
\begin{align*}
\mathcal{A} = & \{(\alpha_1,\dots,\alpha_p) : \alpha_1,\dots,\alpha_p \geq 1, \alpha_i\leq 1+\tau \text{ for at least one}~i, \mbox{ and } \sum_{j:\alpha_j>1} \alpha_j = p + \varepsilon\}
\end{align*}
in time 
$\mathcal{O}\left(T_{\WS_{1+\tau}} \left(\frac{1+\tau}{\varepsilon-\tau\cdot p} \log\frac{\UB}{\LB}\right)^{p-1} \right)$.
\end{corollary}

\begin{proof}
Given $\varepsilon>0$ and $0<\tau<\frac{\varepsilon}{p}$, apply Algorithm~\ref{alg:mainAlgo} with $\varepsilon-\tau\cdot p$ and $\sigma\colonequals 1+\tau$.
\end{proof}

Since any component of a vector in the set~$\mathcal{A}$ from Theorem~\ref{thm:main-result} can get arbitrarily close to $\sigma \cdot p + \varepsilon$ in the worst case, the best ``classical'' approximation result using only a single vector of approximation factors that is obtainable from Theorem~\ref{thm:main-result} reads as follows:

\begin{corollary}\label{cor:classical-result}
Algorithm~\ref{alg:mainAlgo} computes a $(\sigma \cdot p + \varepsilon,\dots,\sigma \cdot p + \varepsilon)$-approximation in time %bounded by $\displaystyle T_{WS_{\sigma}} \sum_{i=1}^p \prod_{j\neq i}\left\lceil \log_{1+\varepsilon}\frac{\UB(j)}{\LB(j)}\right\rceil \in 
$\mathcal{O}\left(T_{WS_{\sigma}} \left(\frac{1}{\varepsilon} \log\frac{\UB}{\LB}\right)^{p-1} \right)$.
\end{corollary}
% \begin{proof}
% The statement follows immediately from Theorem~\ref{thm:main-result} since $\alpha_1,\dots,\alpha_p\geq 1$ and $\sum_{j:\alpha_j>1} \alpha_j \leq \sigma\cdot p$ implies that each~$\alpha_j$ must be less or equal to $\sigma\cdot p$.
% \end{proof}

\subsection{Biobjective Problems}\label{subsec:biobjective}

In this subsection, we focus on biobjective minimization problems. We first specialize some of the general results of the previous subsection to the case $p=2$. Afterwards, we propose a specific approximation algorithm for biobjective problems, which significantly improves upon the running time of Algorithm~\ref{alg:mainAlgo} in the case where an exact algorithm~$\WS_1$ for the weighted sum problem is available.

% Then we propose a specific algorithm, with a running time significantly better than the one obtained by Algorithm~\ref{alg:mainAlgo} with $p=2$. 

\medskip

Theorem~\ref{thm:main-result}, which is the main general result of the previous subsection, can trivially be specialized to the case $p=2$. It is more interesting to consider the situation where the weighted sum can be solved exactly, corresponding to Corollary~\ref{cor:main-result-special}. In that case, we obtain the following result:

\begin{corollary}\label{cor:main-result-special-biobjective}
If $\WS_{\sigma}=\WS_{1}$ is an exact algorithm for the weighted sum problem and $p=2$, Algorithm~\ref{alg:mainAlgo} yields an $\mathcal{A}$-approximation where
\begin{align*}
\mathcal{A} = & \{(1,2 + \varepsilon), (2 + \varepsilon, 1) \}
\end{align*}
in time $\mathcal{O}\left(T_{WS_{1}} \frac{1}{\varepsilon} \log\frac{\UB}{\LB} \right)$.
\end{corollary}

It is worth pointing out that, unlike for the previous results, the set $\mathcal{A}$ of approximation factors is now \emph{finite}. This type of result can be interpreted as a \emph{disjunctive} approximation result: Algorithm~\ref{alg:mainAlgo} outputs a set~$P$ ensuring that, for any $x \in X$, there exists $x' \in P$ such that~$x'$ $(1,2 + \varepsilon)$-approximates~$x$ or~$x'$ $(2 + \varepsilon, 1)$-approximates~$x$.

% Similar to Corollary~\ref{cor:guarantee-by-extreme-sols}, we obtain the following structural result about the approximation guarantee obtained by the set of extreme supported solutions in the biobjective case:

% \begin{corollary}\label{cor:guarantee-by-extreme-sols-biobj}
% For any $\varepsilon>0$, the set of extreme supported solutions of a biobjective problem is an $\mathcal{A}$-approximation, where
% \begin{align*}
% \mathcal{A} = & \{(1,2 + \varepsilon), (2 + \varepsilon, 1) \}.
% \end{align*}
% \end{corollary}

% \subsection{A $\{(1,2+\varepsilon),(2+\varepsilon,1)\}$-approximation for biobjective minimization problems} \label{subsec:biobj_algo}

In the biobjective case, we may scale the weights in the weighted sum problem to be of the form~$(\gamma,1)$ for some $\gamma>0$. In the following, we make use of this observation and refer to a weight vector~$(\gamma,1)$ simply as~$\gamma$. 
%Note that, for a weight vector of this form, every level set of the weighted sum objective function corresponds to a line in~$\mathbb R^2$ with slope~$-\gamma$.

\smallskip

%In Subsection~\ref{biobj_algo}, we present an algorithm  providing a $\{(1,2+\varepsilon),(2+\varepsilon,1)\}$-approximation for biobjective \emph{minimization} problems.

Algorithm~\ref{alg:biobjAlg} is a refinement of Algorithm~\ref{alg:mainAlgo} in the biobjective case when an exact algorithm~$\WS_1$ for the weighted sum problem is available. Algorithm~\ref{alg:mainAlgo} requires to test all the $u_1+u_2+1$ weights $\left(\frac{1}{LB(1)},\frac{1}{LB(2)(1+\varepsilon')^{u_2}} \right)$, $\left(\frac{1}{LB(1)},\frac{1}{LB(2)(1+\varepsilon')^{u_2-1}} \right)$, \ldots, $\left(\frac{1}{LB(1)},\frac{1}{LB(2)} \right)$, $\left(\frac{1}{LB(1)(1+\varepsilon')},\frac{1}{LB(2)} \right)$,\ldots, $\left(\frac{1}{LB(1)(1+\varepsilon')^{u_1}},\frac{1}{LB(2)} \right)$, or equivalently the $u_1+u_2+1$ weights of the form $(\gamma_t,1)$, where $\gamma_t=\frac{LB(2)}{LB(1)}(1+\varepsilon')^{u_2-t+1}$ for $t=1, \ldots, u_1+u_2+1$. Instead of testing all these weights, Algorithm~\ref{alg:biobjAlg} considers only a subset of these weights. More precisely, in each iteration, the algorithm selects a subset of consecutive weights $\{\gamma_{\ell},\ldots,\gamma_r\}$, solves $\WS_{1}(\gamma_t)$ for the weight~$\gamma_t$ with $t=\lfloor\frac{\ell+r}{2}\rfloor$, and decides whether 0, 1, or 2 of the subsets $\{\gamma_{\ell},\ldots,\gamma_t\}$ and $\{\gamma_t,\ldots,\gamma_r\}$ need to be investigated further. This process can be viewed as developing a binary tree where the root, which corresponds to the initialization, requires solving two weighted sum problems, while each other node requires solving one weighted sum problem. This representation is useful to bound the running time of our algorithm. The following technical result on binary trees, whose proof is given in the appendix, will be useful for this purpose:

\begin{lemma}\label{lem:tree-size}
A binary tree with height~$h$ and $k$~nodes with two children contains $\mathcal{O}(k\cdot h)$ nodes.
\end{lemma}

\begin{algorithm2e}
%\begin{algorithm}[!ht]
\SetKw{Compute}{compute}
\SetKw{Break}{break}
\SetKwInOut{Input}{input}\SetKwInOut{Output}{output}
\SetKwComment{command}{right mark}{left mark}

\Input{an instance of a  biobjective minimization problem~$\Pi$, $\varepsilon >0$, an exact algorithm $\WS_{1}$ for the weighted sum problem}

\Output{a $\{(1,2+\varepsilon),(2+\varepsilon,1)\}$-approximation~$P$ for problem~$\Pi$}

\BlankLine
$\varepsilon' \leftarrow \frac{\varepsilon}{p}$\\
$u_1 \leftarrow \mbox{largest~integer~ such~that~} LB(1) (1+\varepsilon')^{u_1}\leq UB(1)$\\
$u_2 \leftarrow \mbox{largest~integer~ such~that~} LB(2) (1+\varepsilon')^{u_2}\leq UB(2)$\\
$\gamma_1\leftarrow \frac{LB(2)}{\LB(1)} (1+\varepsilon')^{u_2}$;\qquad
$\gamma_{u_1+u_2+1}\leftarrow \frac{LB(2)}{\LB(1)} (1+\varepsilon')^{-u_1} $\\
$x^1 \leftarrow \WS_{1}(\gamma_1)$ ;\qquad
$x^{u_1+u_2+1} \leftarrow \WS_{1}(\gamma_{u_1+u_2+1})$\\
 $Q \leftarrow  \emptyset$\\
\lIf{$x^1$  $(1,2+\varepsilon)$-approximates~$x^{u_1+u_2+1}$} {$P \leftarrow \{x^1\}$}
\Else{\lIf{$x^{u_1+u_2+1}$ $(2+\varepsilon,1)$-approximates~$x^1$}{$P \leftarrow \{x^{u_1+u_2+1}\}$}
	\Else{$P \leftarrow \{x^1,x^{u_1+u_2+1}\}$\\
       	   $Q \leftarrow \{(1,u_1+u_2+1)\}$}
	}
\While{$Q \neq \emptyset$}{
	Select $(\ell,r)$ from~$Q$\\
    $Q \leftarrow Q \setminus \{(\ell,r)\}$\\
    $t \leftarrow \lfloor \frac{\ell+r}{2} \rfloor $\\
    $\gamma_t=\frac{LB(2)}{LB(1)}(1+\varepsilon')^{u_2-t+ 1}$ \\
     $x^t \leftarrow \WS_{1}(\gamma_t)$
     
   \If{$x^{\ell}$ does not $(1,2+\varepsilon)$-approximate~$x^{t}$ or~$x^{r}$ does not $(2+\varepsilon,1)$-approximate~$x^t$  \label{if1} }{
        	$P\leftarrow P\cup \{x^t\}$\\
            \If{$t\geq \ell+2$ and $x^{\ell}$ does not $(1,2+\varepsilon)$-approximate~$x^{t}$ and ~$x^{t}$ does not $(2+\varepsilon,1)$-approximate~$x^{\ell}$ \label{if2}} 				{$Q\leftarrow Q\cup \{(\ell,t)\}$ \label{if2then}}
        		\If{$t\leq r-2$ and $x^{t}$ does not $(1,2+\varepsilon)$-approximate~$x^{r}$ and ~$x^{r}$ does not $(2+\varepsilon,1)$-approximate~$x^{t}$ \label{if3}}   				{$Q\leftarrow Q\cup \{(t,r)\}$ \label{if3then}}
            }
     }
\Return $P$
\caption{A \emph{$\{(1,2+\varepsilon),(2+\varepsilon,1)\}$-approximation for  biobjective minimization problems}\label{alg:biobjAlg}}
\end{algorithm2e}

\begin{theorem}\label{thm:binary-sear-approx} For a  biobjective minimization problem, Algorithm~\ref{alg:biobjAlg} returns a $\{(1,2+\varepsilon),(2+\varepsilon,1)\}$-approximation in time
\begin{align*} 
\mathcal{O}\left(T_{\WS_1}\cdot\log \left(\frac{1}{\varepsilon}\cdot\log \frac{\UB}{\LB}\right) \cdot\log \frac{\UB}{\LB}  \right).
\end{align*}
\end{theorem}

\begin{proof}
The approximation guarantee of Algorithm~\ref{alg:biobjAlg} derives from Theorem~\ref{thm:main-result}. We just need to prove that the subset of weights used here is sufficient to preserve the approximation guarantee. 

\smallskip

In lines~\ref{if2}-\ref{if2then}, the weights~$\gamma_i$ for $i=\ell+1, \ldots t-1$ are not considered if $x^{\ell}$  $(1,2+\varepsilon)$-approximates~$x^{t}$ or if ~$x^t$ $(2+\varepsilon,1)$-approximates~$x^{\ell}$. We show that, indeed, these weights are not needed. 

\smallskip

To this end, first observe that any solution $x^i \colonequals \WS_{1}(\gamma_i)$ for $i\in\{\ell+1, \ldots t-1\}$ is such that 
\begin{align*}
	f_1(x^{\ell}) \leq f_1(x^i) \leq f_1(x^{t})\phantom{.} \quad & \text{ and } \\
    f_2(x^{\ell}) \geq f_2(x^i) \geq f_2(x^{t}). \quad & \quad
\end{align*}
since $\gamma_{\ell} > \gamma_i >\gamma_t$. Thus, if ~$x^{\ell}$ $(1,2+\varepsilon)$-approximates~$x^{t}$, we obtain 
\begin{align*}
	   f_2(x^{\ell}) \leq (2+\varepsilon)\cdot f_2(x^t) \leq (2+\varepsilon)\cdot f_2(x^i), \quad  
\end{align*}
which shows that~$x^{\ell}$ also $(1,2+\varepsilon)$-approximates~$x^i$. Therefore,~$x^i$ and the corresponding weight~$\gamma_i$ are not needed.

\smallskip
 
\noindent
Similarly, if~$x^t$ $(2+\varepsilon,1)$-approximates~$x^{\ell}$, we have
\begin{align*}
	f_1(x^t) \leq (2+\varepsilon)\cdot f_1(x^{\ell}) \leq (2+\varepsilon)\cdot f_1(x^i), \quad  
\end{align*}
which shows that~$x^t$ $(2+\varepsilon,1)$-approximates~$x^i$. Therefore~$x^i$ and the corresponding weight~$\gamma_i$ are again not needed.

\smallskip

In lines~\ref{if3}-\ref{if3then}, the weights~$\gamma_i$ for $i=t+1, \ldots, r-1$ are not considered if~$x^t$ $(1,2+\varepsilon)$-approximates~$x^{r}$ or if~$x^r$ $(2+\varepsilon,1)$-approximates~$x^{t}$ for similar reasons.

\smallskip

Also, in line~\ref{if1},~$x^t$ can be discarded and the weights $\gamma_i$ for $i=\ell+1, \ldots, r-1$ can be ignored if~$x^{\ell}$ $(1,2+\varepsilon)$-approximates~$x^{t}$ and~$x^r$ $(2+\varepsilon,1)$-approximates~$x^{t}$. Indeed, using similar arguments as before, we obtain that $x^{\ell}$ $(1,2+\varepsilon)$-approximates~$x^{i}$ for $i=\ell+1, \ldots, t$ and~$x^{r}$ $(2+\varepsilon, 1)$-approximates~$x^{i}$ for $i=t, \ldots, r-1$ in this case. Consequently, compared to Algorithm~\ref{alg:mainAlgo}, only superfluous weights are discarded in Algorithm~\ref{alg:biobjAlg} and the approximation guarantee follows by Theorem~\ref{thm:main-result}.

\bigskip

We now prove the claimed bound on the running time. Algorithm~\ref{alg:biobjAlg} explores a set of weights of cardinality $u_1+u_2+1$ = $\left\lfloor \log_{1+\varepsilon'}\frac{\UB(1)}{\LB(1)}\right\rfloor + \left\lfloor \log_{1+\varepsilon'}\frac{\UB(2)}{\LB(2)}\right\rfloor + 1$. The running time is obtained by bounding the number of calls to algorithm~$\WS_{1}$, which corresponds to the number of nodes of the binary tree implicitly developed by the algorithm. The height of this tree is $\log_2(u_1+u_2+1) \in \mathcal{O}\left(\log \left(\frac{1}{\varepsilon}\cdot\log \frac{\UB}{\LB}\right)\right)$.

% \smallskip

% The height of the tree can be upper bounded as before: Since the algorithm stops bisecting an interval if the size of the remaining subintervals drops below $\frac{1}{\UB(1)^2}$ and the algorithm starts with the interval $[\gamma_1, \gamma_2]=[\frac{1}{\UB(1)},\UB(2)]$, there can be at most
% \begin{align*}
% \log_2 \left(\left(\UB(2)-\frac{1}{\UB(1)}\right)\cdot \UB(1)^2\right)
% \leq 2\log_2\left(\UB(1)\cdot\UB(2)\right)
% \end{align*}
% consecutive subdivisions, which upper bounds the height of the tree.

\smallskip

In order to bound the number of nodes with two children in the tree, we observe that we generate such a node (i.e. add the pairs $(\ell,t)$ and $(t,r)$ to~$Q$) only if~$x^\ell$ does not $(1,2+\varepsilon)$-approximate~$x^t$ and $x^t$ does not $(2+\varepsilon,1)$-approximate~$x^{\ell}$, and also~$x^t$ does not $(1,2+\varepsilon)$-approximate~$x^r$ and~$x^r$ does not $(2+\varepsilon,1)$-approximate~$x^t$. Hence, whenever a node with two children is generated, the corresponding solution~$x^t$ does neither $(1,2+\varepsilon)$ nor $(2+\varepsilon,1)$-approximate any previously generated solution and vice versa, so their objective values in both of the two objective functions must differ by more than a factor~$(2+\varepsilon)$. Using that the $j$th objective value of any feasible solution is between~$\LB(j)$ and~$\UB(j)$, this implies that there can be at most
\begin{align*}
\min\left\{\log_{2+\varepsilon} \left(\frac{\UB(1)}{\LB(1)}\right) ; \log_{2+\varepsilon} \left( \frac{\UB(2)}{\LB(2)}\right) \right\} \in \mathcal{O}\left(\log \frac{\UB}{\LB}  \right)
\end{align*}
nodes with two children in the tree.

\smallskip

Using the obtained bounds on the height of the tree and the number of nodes with two children, Lemma~\ref{lem:tree-size} shows that the total number of nodes in the tree is
\begin{align*}
	\mathcal{O}\left(\log \left(\frac{1}{\varepsilon}\cdot\log \frac{\UB}{\LB}\right) \cdot\log \frac{\UB}{\LB}  \right),
\end{align*}
which proves the claimed bound on the running time.
\end{proof}

\subsection{Tightness results}\label{subsec:tightness}

When solving the weighted sum problem exactly, Co\-rol\-la\-ry~\ref{cor:main-result-special} states that Algorithm~\ref{alg:mainAlgo} obtains a set~$\mathcal{A}$ of approximation factors in which $\sum_{j:\alpha_j>1}\alpha_j=p+\varepsilon$ for each $\alpha=(\alpha_1,\dots,\alpha_p)\in\mathcal{A}$. 

The following theorem shows that this multi-factor approximation result is arbitrarily close to the best possible result obtainable by supported solutions:

\begin{theorem}\label{thm:inapprox-min}
For $\varepsilon>0$, let
\begin{align*}
	\mathcal{A}\colonequals\{\alpha\in\mathbb{R}^p: \alpha_1,\dots,\alpha_p\geq 1,\; \alpha_i=1 \text{ for at least one}~i, \text{ and } \sum_{j:\alpha_j>1}\alpha_j=p-\varepsilon\}.
\end{align*}
Then there exists an instance of a $p$-objective minimization problem for which the set~$\XS$ of supported solutions is not an $\mathcal{A}$-approximation.
% Then there exists an instance of a $p$-objective minimization problem for which at least one unsupported solution is not $\alpha$-approximated by any supported solution for any $\alpha\in\mathcal{A}'$.
\end{theorem}

\begin{proof}
In the following, we only specify the set~$Y$ of images. A corresponding instance consisting of a set~$X$ of feasible solutions and an objective function~$f$ can then easily be obtained, e.\,g., by setting $X\colonequals Y$ and $f\colonequals \text{id}_{\mathbb{R}^p}$.

\smallskip

For $M>0$, let $Y\colonequals\{y^{1},\dots,y^{p},\tilde{y}\}$ with $y^{1}=(M,\frac{1}{p},\dots,\frac{1}{p})$, $y^{2}=(\frac{1}{p},M,\frac{1}{p},\dots,\frac{1}{p})$, \dots, $y^{p}=(\frac{1}{p},\dots,\frac{1}{p},M)$ and $\tilde{y}=\left(\frac{M+1}{p},\dots,\frac{M+1}{p}\right)$. Note that the point~$\tilde{y}$ is unsupported, while $y^{1},\dots,y^{p}$ are supported (an illustration for the case $p=2$ is provided in Figure~\ref{fig:minimization-impossiblility}).

\medskip

\noindent
Moreover, the ratio of the $j$-th components of the points~$y^j$ and~$\tilde{y}$ is exactly
\begin{align*}
	\frac{M}{\nicefrac{(M+1)}{p}} = p\cdot\frac{M}{M+1},
\end{align*}
which is larger than $p-\varepsilon$ for $M>\frac{p}{\varepsilon}-1$. Consequently, for such~$M$, the point~$\tilde{y}$ is not $\alpha$-approximated by any of the supported points $y^1,\dots,y^p$ for any $\alpha\in\mathcal{A}$, which proves the claim.
\end{proof}

\begin{figure}[ht!]
	\begin{center}
		\begin{tikzpicture}[scale=1.25]
		\fill[gray!30] (0.05,3.421) -- (0.05,7.4) -- (7.4,7.4) -- (7.4,3.421);
		\fill[gray!30] (3.421,0.05) -- (3.421,7.4) -- (7.4,7.4) -- (7.4,0.05);
    \fill[gray!30] (0.0263,7.4) -- (0.05,7.4) -- (0.05,6.5) -- (0.0263,6.5);
    \fill[gray!30] (7.4,0.0263) -- (7.4,0.05) -- (6.5,0.05) -- (6.5,0.0263);
		\draw[->] (-0.2,0) -- (7.4,0) node[below right] {$f_1$};
		\draw[->] (0,-0.2) -- (0,7.4) node[above left] {$f_2$};
		\draw[dashed,gray] (0.05,6.5) -- (6.5,0.05);
		\fill (3.34,3.34) circle (1.5pt) node[above right] {$\tilde{y}=(\frac{M+1}{2},\frac{M+1}{2})$};
    \fill (0.05,6.5) circle (1.5pt) node[above right] {$y^2=(\frac{1}{2},M)$};
    \fill (6.5,0.05) circle (1.5pt) node[above right] {$y^1=(M,\frac{1}{2})$};
		\end{tikzpicture}
		\caption{Image space of the instance constructed in the proof of Theorem~\ref{thm:inapprox-min} for $p=2$. The shaded region is $\{(1,2-\varepsilon),(2-\varepsilon,1)\}$-approximated by the supported points~$y^1,y^2$.}\label{fig:minimization-impossiblility}
     % \varepsilon=0.1 in the picture and all points were scaled down by a factor of 10
	\end{center}
\end{figure}
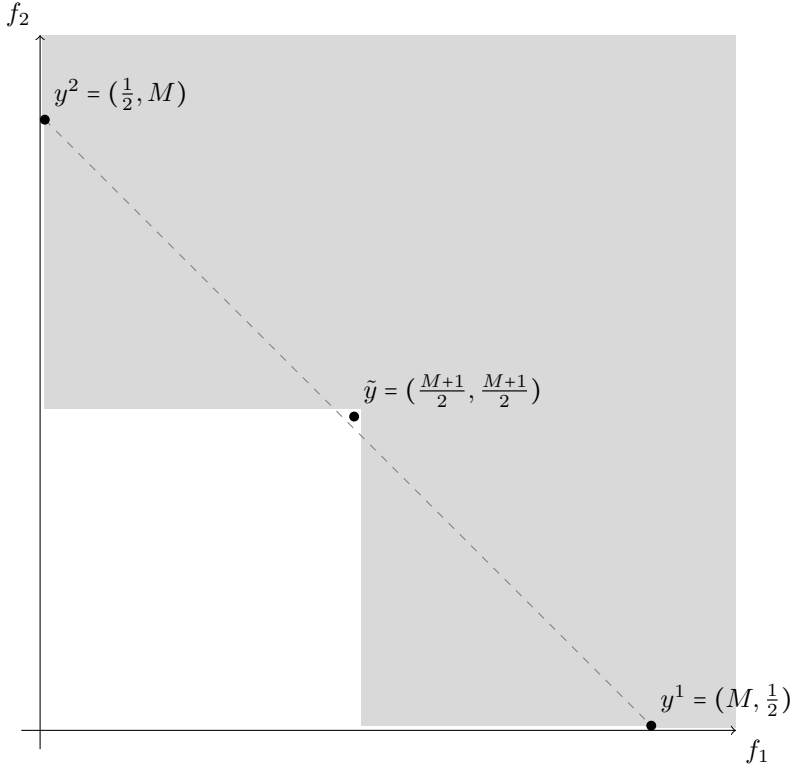

% \begin{figure}[ht!]
% 	\begin{center}
% 		\begin{tikzpicture}[scale=1.5]
% 		\fill[gray!30] (0.5,3.25) -- (0.5,7.4) -- (7.4,7.4) -- (7.4,3.25);
% 		\fill[gray!30] (3.25,0.5) -- (3.25,7.4) -- (7.4,7.4) -- (7.4,0.5);
% 		\draw[->] (-0.2,0) -- (7.4,0) node[below right] {$f_1$};
% 		\draw[->] (0,-0.2) -- (0,7.4) node[above left] {$f_2$};
% 		\fill (0.5,6.5) circle (2pt) node[above right] {$y^2=(\frac{1}{2},M)$};
% 		\fill (6.5,0.5) circle (2pt) node[above right] {$y^1=(M,\frac{1}{2})$};
% 		\draw[dashed,gray] (0.5,6.5) -- (6.5,0.5);
% 		\fill (3.75,3.75) circle (2pt) node[above right] {$\tilde{y}=(\frac{M+1}{2},\frac{M+1}{2})$};
% 		\end{tikzpicture}
% 		\caption{Image space of a biobjective minimization problem with three points, only two of which are supported (where $M>0$ is a large integer). Each of the two supported points does not approximate the unsupported point~$\tilde{y}$ better than a factor~$M$ in one of the two objective functions. The region that is $\{(1,2),(2,1)\}$-approximated by the supported points~$y^1,y^2$ is highlighted.}\label{fig:minimization-impossiblility}
% 	\end{center}
% \end{figure}

% Note that the requirement that $\alpha_i=1$ for at least one~$i$ in the definition of the set~$\mathcal{A}$ of approximation factors is not important for the inapproximability result obtained in Theorem~\ref{thm:inapprox-min}. As the proof shows, the same result holds when omitting this requirement from the definition of~$\mathcal{A}$.

We remark that the set of points~$Y$ constructed in the proof of Theorem~\ref{thm:inapprox-min} can easily be obtained from instances of many well-known multiobjective minimization problems such as multiobjective shortest path, multiobjective spanning tree, multiobjective minimum ($s$-$t$-) cut, or multiobjective TSP (for multiobjective shortest path, for example, a collection of $p+1$ disjoint $s$-$t$-paths whose cost vectors correspond to the points $y^{1},\dots,y^{p},\tilde{y}$ suffices). Consequently, the result from Theorem~\ref{thm:inapprox-min} holds for each of these specific problems as well.

\medskip

Moreover, note that also the classical approximation result obtained in Corollary~\ref{cor:classical-result} is arbitrarily close to best possible in case that the weighted sum problem is solved exactly: While Corollary~\ref{cor:classical-result} shows that a $(p+\varepsilon,\dots,p+\varepsilon)$-approximation is obtained from Algorithm~\ref{alg:mainAlgo} when solving the weighted sum problem exactly, the instance constructed in the proof of Theorem~\ref{thm:inapprox-min} shows that the supported solutions do not yield an approximation guarantee of $(p-\varepsilon,\dots,p-\varepsilon)$ for any $\varepsilon>0$. This yields the following theorem:

\enlargethispage{\baselineskip}

\begin{theorem}\label{thm:inapprox-min-classical}
For any $\varepsilon>0$, there exists an instance of a $p$-objective minimization problem for which the set~$\XS$ of supported solutions is not a $(p-\varepsilon,\dots,p-\varepsilon)$-approximation.
% Then there exists an instance of a $p$-objective minimization problem for which at least one unsupported solution is not $\alpha$-approximated by any supported solution.
\end{theorem}

\section{Applications}\label{sec:applications}

Our results can be applied to a large variety of minimization problems since exact or approximate polynomial-time algorithms are available for the weighted sum scalarization of many problems.

\subsection{Problems with a polynomial-time solvable weighted sum scalarization}

If the weighted sum scalarization can be solved exactly in polynomial time, Corollary~\ref{cor:main-result-special} shows that Algorithm~\ref{alg:mainAlgo} yields a multi-factor approximation where each feasible solution is approximated with some approximation guarantee $(\alpha_1,\dots,\alpha_p)$ such that $\sum_{j:\alpha_j>1}\alpha_j=p+\varepsilon$ and $\alpha_i=1$ for at least one~$i$.

\smallskip

Many problems of this kind admit an MFPTAS, i.e., a $(1+\varepsilon,\dots,1+\varepsilon)$-appro\-xi\-ma\-tion that can be computed in time polynomial in the encoding length of the input and $\frac{1}{\varepsilon}$. The approximation guarantee we obtain is worse in this case, even if the sum of the approximation factors for which an error can be observed is $p+\varepsilon$ in both approaches. The running time, however, is usually significantly better in our approach.

\smallskip

For the multiobjective shortest path problem, for example, the existence of an MFPTAS was shown in~\cite{Papadimitriou+Yannakakis:multicrit-approx}, while several specific MFPTAS have been proposed. Among these, the MFPTAS with the best running time is the one proposed in~\cite{Tsaggouris+Zaroliagis:mult-shortest-path}. For $p\geq 2$, their running time for general digraphs with $n$~vertices and $m$~arcs is $\mathcal{O}\left(m\cdot n^p \left(\frac{1}{\varepsilon}\log\frac{\UB}{\LB} \right)^{p-1}\right)$ while ours is only $\mathcal{O}\left((m + n\log\log n) \cdot \left(\frac{1}{\varepsilon}\log\frac{\UB}{\LB} \right)^{p-1}\right)$ using one of the fastest algorithms for single objective shortest path~\cite{Thorup:SP}, and even $\mathcal{O}\left((m + n\log\log n) \cdot \log \left(\frac{1}{\varepsilon}\log\frac{\UB}{\LB} \right)\cdot\log\frac{\UB}{\LB}\right)$ for $p=2$, using Theorem~\ref{thm:binary-sear-approx} and the same single objective algorithm.

\smallskip

There are, however, also problems for which the weighted sum scalarization can be solved exactly in polynomial time, but whose multiobjective version does \emph{not} admit an MFPTAS unless $\textsf{P}=\textsf{NP}$. For example, this is the case for the minimum $s$-$t$-cut problem~\cite{Papadimitriou+Yannakakis:multicrit-approx}. For yet other problems, like, e.g., the minimum weight perfect matching problem, only a randomized MFPTAS is known so far~\cite{Papadimitriou+Yannakakis:multicrit-approx}. In both cases, our algorithm can still be applied.

\subsection{Problems with a polynomial-time approximation scheme for the weighted sum scalarization}

For problems where the weighted sum scalarization admits a polynomial-time approximation scheme, Corollary~\ref{cor:weighted-sum-PTAS} shows that Algorithm~\ref{alg:mainAlgo} yields a multi-factor approximation where each feasible solution is approximated with some approximation guarantee $(\alpha_1,\dots,\alpha_p)$ such that $\sum_{j:\alpha_j>1}\alpha_j=p+\varepsilon$. Thus, only the property that $\alpha_i=1$ for at least one~$i$ is lost compared to the case where the weighted sum scalarization can be solved exactly in polynomial time. 

\smallskip

Since there exists a vast variety of single objective problems that admit po\-ly\-no\-mi\-al-time approximation schemes, this result is also widely applicable and yields the best known multiobjective approximation results for many problems. For example, we obtain the best known approximation results for the multiobjective versions of the weighted planar TSP (for which a polynomial-time approximation scheme with running time linear in the number~$n$ of vertices exists~\cite{Klein:planar-TSP}) and minimum weight planar vertex cover (for which a polynomial-time approximation scheme was proposed in~\cite{Baker:planar-graphs}). Note that, for both of these problems and many others in this class, it is not known whether the multiobjective version admits an MPTAS.

\subsection{Problems with a polynomial-time $\sigma$-approximation for the weigh\-ted sum scalarization}

If the weighted sum scalarization admits a polynomial-time $\sigma$-approxi\-ma\-tion algorithm (where $\sigma$ can be either a constant or a function of the input size), Theorem~\ref{thm:main-result} shows that Algorithm~\ref{alg:mainAlgo} yields a multi-factor approximation where each feasible solution is approximated with some approximation guarantee $(\alpha_1,\dots,\alpha_p)$ such that $\sum_{j:\alpha_j>1}\alpha_j=\sigma\cdot p+\varepsilon$ and $\alpha_i\leq\sigma$ for at least one~$i$. Moreover, by Corollary~\ref{cor:classical-result}, the algorithm also yields a (classical) $(\sigma\cdot p+\varepsilon,\dots,\sigma\cdot p+\varepsilon)$-approximation.

\smallskip

These results yield the best known approximation guarantees for many well-studied problems whose single objective version does not admit a po\-ly\-no\-mi\-al-time approximation scheme unless $\textsf{P}=\textsf{NP}$. Consequently, the multiobjective version of these problems does not admit an MPTAS under the same assumption. Problems of this kind include, e.g., minimum weight vertex cover, minimum $k$-spanning tree, minimum weight edge dominating set, and minimum metric $k$-center, all of which admit $2$-approximation algorithms in the single objective case (see \cite{Bar-Yehuda+Even:vertex-cover,Garg:STOC05,Fujito+Nagamochi:edge-dom-set,Hochbaum+Shmoys:bottleneck-problems}, respectively). An example of a problem where only a non-constant approximation factor can be obtained in the single objective case is the minimum weight set cover problem, where only a $(1+\ln |S|)$-approximation exists with~$|S|$ denoting the cardinality of the ground set to cover~\cite{Chvatal:set-cover-approx}. For all of these problems, Algorithm~\ref{alg:mainAlgo} yields both the best known classical approximation result for the multiobjective version as well as the first multi-factor approximation result.

\smallskip

A particularly interesting problem of this class is the metric version of the symmetric traveling salesman problem (metric STSP). Here, problem specific (deterministic) algorithms exist that obtain approximation guarantees of $(2,2)$ in the biobjective case~\cite{Glasser+etal:TSP-generalized} and $(2+\varepsilon,\dots,2+\varepsilon)$ for any constant number of objectives~\cite{Manthey+Ram:multicriteria-TSP}. The best approximation algorithm for the single objective version is the $\frac{3}{2}$-approximation algorithm by Christofides~\cite{Christofides:TSP}, which can be used in Algorithm~\ref{alg:mainAlgo} in order to obtain a multi-factor approximation where each feasible solution is approximated with some approximation guarantee $(\alpha_1,\dots,\alpha_p)$ such that $\sum_{j:\alpha_j>1}\alpha_j=\frac{3}{2}\cdot p+\varepsilon$ and $\alpha_i\leq\frac{3}{2}$ for at least one~$i$. 

\section{An inapproximability result for maximization problems}\label{sec:maximixation}

In this section, we show that the weighted sum scalarization is much less powerful for approximating multiobjective \emph{maximization} problems.

\smallskip

Intuitively, in a multiobjective \emph{minimization} problem, the positivity of the objective function values and of the weights used within the weighted sum scalarization implies that a bad (i.e., large) value in some of the objective functions cannot be compensated in the weighted sum by a good (i.e., small) value in another objective function. This means that, if the weighted sum of the objective values of a solution for a minimization problem is (close to) optimal (i.e., minimal), then no single objective value can be too large. More precisely, for $f(x)=(f_1(x), \dots , f_p(x))\in \mathbb R^p$ and $w\in\mathbb{R}^p$ with $f_j(x),w_j > 0$ for $j=1, \dots , p$ and $v > 0$, it holds that
\begin{align*}
\sum_{j=1}^p w_j f_j(x) \leq v \Longrightarrow  f_j(x) \leq \frac{1}{w_j} v  \text{ for all } j=1, \dots , p.
\end{align*}

For a \emph{maximization} problem, however, a bad (i.e., small) value in some of the objective functions can be completely compensated in the weighted sum by a very good (i.\,e., large) value in another objective function. Thus, a solution that obtains a (close to) optimal (i.e., maximal) value in the weighted sum of the objective values can still have a very small (i.\,e., bad) value in some of the objectives:
\begin{align*}
\sum_{j=1}^p w_jf_j(x) \geq v \not\Longrightarrow  f_j(x) \geq cv \text{ for all } j=1, \dots , p \text{ and any constant } c>0.
\end{align*}

For instance, while Corollary~\ref{cor:main-result-special-biobjective} implies that the set of supported solutions yields a $\{(1,2+\varepsilon),(2+\varepsilon,1)\}$-approximation for any $\varepsilon>0$ in the case of a biobjective minimization problem, no similar result holds for maximization problems. Indeed, for biobjective maximization problems, Figure~\ref{fig:maximization-impossiblility} demonstrates that there may exist unsupported solutions that are approximated only with an arbitrarily large approximation factor in (all but) one objective function by any supported solution. 

\begin{figure}[ht!]
	\begin{center}
		\begin{tikzpicture}[scale=1.25]
		\fill[gray!30] (1,6.5) -- (0,6.5) -- (0,0) -- (1,0) -- (1,6.5);
		\fill[gray!30] (6.5,1) -- (6.5,0) -- (0,0) -- (0,1) -- (6.5,1);
    \fill[gray!30] (0.5,7.4) -- (0,7.4) -- (0,0) -- (0.5,0) -- (0.5,7.4);
    \fill[gray!30] (7.4,0.5) -- (7.4,0) -- (0,0) -- (0,0.5) -- (7.4,0.5);
		\draw[->] (-0.2,0) -- (7.4,0) node[below right] {$f_1$};
		\draw[->] (0,-0.2) -- (0,7.4) node[above left] {$f_2$};
		\fill (0.5,6.5) circle (1.5pt) node[above right] {$y^2=(\frac{1}{2},M)$};
		\fill (6.5,0.5) circle (1.5pt) node[above right] {$y^1=(M,\frac{1}{2})$};
		\draw[dashed,gray] (0.5,6.5) -- (6.5,0.5);
		\fill (3.25,3.25) circle (1.5pt) node[below left] {$\tilde{y}=(\frac{M}{2},\frac{M}{2})$};
		\end{tikzpicture}
		\caption{Image space of a biobjective maximization problem with three points, only two of which are supported (where $M>0$ is a large integer). Each of the two supported points does not approximate the unsupported point~$\tilde{y}$ better than a factor~$M$ in one of the two objective functions. The shaded region is $\{(1,2),(2,1)\}$-approximated by the supported points~$y^1,y^2$.}\label{fig:maximization-impossiblility}
	\end{center}
\end{figure}
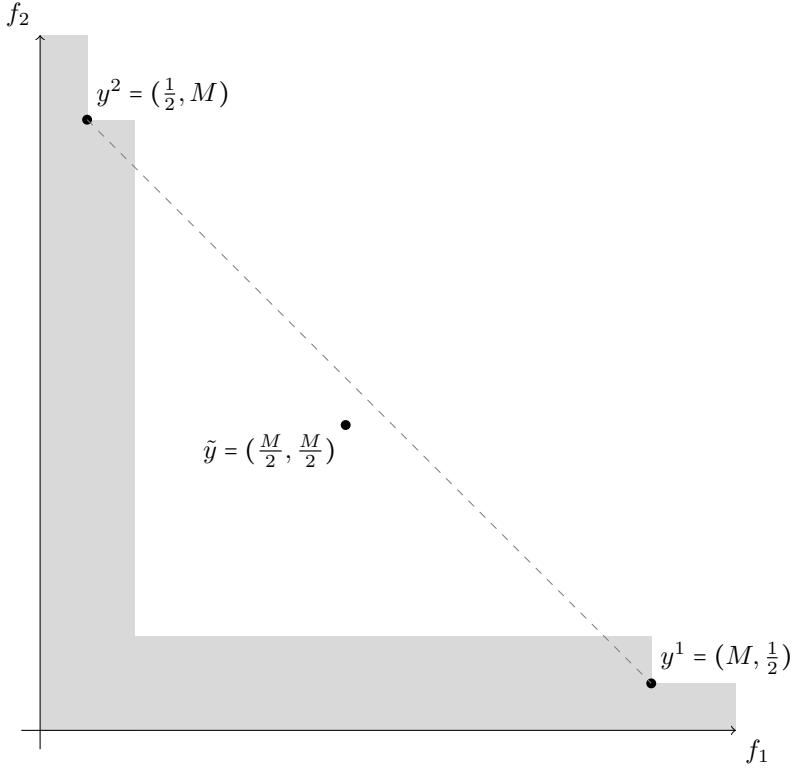

\medskip

The following theorem generalizes the construction in Figure~\ref{fig:maximization-impossiblility} to an arbitrary number of objectives and shows that, for maximization problems, a polynomial approximation factor can, in general, not be obtained in more than one of the objective functions simultaneously even if the approximating set consists of all the supported solutions:

\begin{theorem}\label{thm:maximization-impossiblility}
For any $p\geq 2$ and any polynomial~$\text{pol}$, there exists an instance~$I$ of a $p$-objective maximization problem such that at least one unsupported solution is \emph{not} approximated with an approximation guarantee of $2^{\text{pol}(|I|)}$ in~$p-1$ of the objective functions by any supported solution.
\end{theorem}

\begin{proof}
Given $p\geq 2$ and a polynomial~$\text{pol}$, consider the $p$-objective maximization problem where each instance~$I$ is given by a $(p+1)$-tuple $(x^1,\dots,x^p,\tilde{x})$ of pairwise different vectors $x^1,\dots,x^p,\tilde{x}\in\mathbb{Z}^p$ and the feasible set is $X=\{x^1,\dots,x^p,\tilde{x}\}$.

Given the encoding length~$|I|$ of such an instance%\footnote{\textcolor{blue}{Note that neither the objective values nor the (polynomially computable) objective functions~$f_i$ have to be encoded as part of an instance of a multiobjective optimization problem, cf.~~\cite{Papadimitriou+Yannakakis:multicrit-approx}.}}
, we set $M\colonequals 2^{\text{pol}(|I|)} + 1$ and $f(x^{1})=(M,\frac{1}{p},\dots,\frac{1}{p})$, $f(x^{2})=(\frac{1}{p},M,\frac{1}{p},\dots,\frac{1}{p})$, \dots, $f(x^{p})=(\frac{1}{p},\dots,\frac{1}{p},M)$ and $f(\tilde{x})=\left(\frac{M}{p},\dots,\frac{M}{p}\right)$. Then, the solution~$\tilde{x}$ is unsupported, while $x^{1},\dots,x^{p}$ are supported.

\smallskip

Moreover, the ratio of the $j$-th components of the images~$f(\tilde{x})$ and~$f(x^{\ell})$ for any $j\neq \ell$ is exactly $M=2^{\text{pol}(|I|)} + 1 > 2^{\text{pol}(|I|)}$, which shows that~$x^{\ell}$ does not yield an approximation guarantee of $2^{\text{pol}(|I|)}$ in objective function~$f_j$ for any $j\neq \ell$.\qed
\end{proof}

\section{Conclusion}\label{sec:conclusion}

The weighted sum scalarization is the most frequently used method to transform a multiobjective into a single objective optimization problem. In this article, we contribute to a better understanding of the quality of approximations for general multiobjective optimization problems which rely on this scalarization technique.
To this end, we refine and extend the common notion of approximation quality in multiobjective optimization. As we show, the resulting multi-factor notion of approximation more accurately describes the approximation quality in multiobjective contexts.
We also present an efficient approximation algorithm for general multiobjective minimization problems which turns out to be best possible under some additional assumptions. Interestingly, we show that a similar result based on supported solutions cannot be obtained for multiobjective maximization problems.

\medskip

The new multi-factor notion of approximation is independent of the specific algorithms used here. Thus, a natural direction for future research is to analyze new and existing approximation algorithms more precisely with the help of this new notion. This may yield both a better understanding of existing approaches as well as more accurate approximation results.

\appendix
\section{Proof of Lemma~\ref{lem:tree-size}}

\begin{proof}
In order to show the claimed upper bound on the number of nodes, we first show that any binary tree~$T$ with height~$h$ and $k$~nodes with two children that has the maximum possible number of nodes among all such binary trees must have the following property: If~$v$ is a node with two children at level~$\ell$, then all nodes~$u$ at the levels~$0,\dots,\ell-1$ must also have two children.

So assume by contradiction that~$T$ is a binary tree maximizing the number of nodes among all trees with height~$h$ and $k$~nodes with two children, but~$T$ does not have this property. Then there exists a node~$v$ in~$T$ with two children at some level~$\ell$ and a node~$u$ with at most one child at a lower level~$\ell'\in\{0,\dots,\ell-1\}$. Then, the binary tree~$T'$, that results from making one node~$w$ that is a child of~$v$ in~$T$ an (additional) child of~$u$, also has height~$h$, contains~$k$ nodes with two children, and has the same number of nodes as~$T$. Moreover, the level of~$w$ in~$T'$ changes to $\ell'+1<\ell+1$. Hence, also the level of any leave of the subtree rooted at~$w$ must have decreased by at least one. Thus, giving any leave of this subtree an additional child in~$T'$ would yield a binary tree of height~$h$ and~$k$ nodes with two children, and a strictly larger number of nodes than~$T$, contradicting the maximality of~$T$.

\smallskip

By the above property, in any binary tree maximizing the number of nodes among the trees satisfying the assumptions of the lemma, there are only nodes with two children on all levels $i<h'\colonequals \lfloor\log_2(k+1)\rfloor$ and only nodes with at most one child on all levels $i>h'$. Level~$h'$ may contain nodes with two children, but there is at least one node with only one child on this level.

\enlargethispage{2.5\baselineskip}

Consequently, there are at most~$k$ nodes in total on the levels $0,\dots,h'-1$ and at most $k+1$ nodes at level~$h'$. Moreover, there are at most $2(k+1)$ nodes at level~$h'+1$, each of which is the root of a subtree (path) consisting of at most~$h-h'$ nodes (each with at most one child). Overall, this proves an upper bound of at most $k + (k+1) + 2(k+1)\cdot(h-h')\in\mathcal{O}(k\cdot h)$ on the number of nodes in the tree.
% If we denote the number of nodes on level~$h'$ by~$n(h')$, then, since the tree is a complete binary tree up to level~$h'-1$, it follows that there are exactly $n(h')-1$ nodes (with two children each) at the levels $0,\dots,h'-1$ in total. Consequently, since there is at least one node with only one child on level~$h'$ and there are only~$k$ nodes with two children in the tree, there can be at most $\frac{k}{2}$ nodes with two children at level~$h'$.
% The claim then follows since there can then be at most~$\frac{k}{2}$ nodes on level~$h'$ with two children, each of which is the root of a subtree (path) consisting of at most~$h-h_2$ nodes (each with at most one child).
~\end{proof}

%\section*{\refname}

\newpage

\bibliographystyle{siamplain}
\bibliography{literature}
\end{document}